\documentclass[11pt]{article}

\usepackage[margin=1in]{geometry}
\usepackage[T1]{fontenc}
\usepackage{amsmath,amssymb,amsthm}
\usepackage{mathtools}
\usepackage{graphicx}
\usepackage{natbib}
\usepackage{microtype}
\microtypesetup{expansion=false}
\usepackage{placeins}
\usepackage{float}
\usepackage{array}
\usepackage{url}
\usepackage[hidelinks]{hyperref}
\hypersetup{
  pdftitle={Screening-Off Information and Conditional Risk in Portfolio Choice},
  pdfauthor={Alejandro Rodriguez Dominguez},
  pdfsubject={Original Article submitted to Quantitative Finance}
}
\newcolumntype{P}[1]{>{\raggedright\arraybackslash}p{#1}}

\theoremstyle{plain}
\newtheorem{theorem}{Theorem}[section]
\newtheorem{proposition}[theorem]{Proposition}
\newtheorem{corollary}[theorem]{Corollary}

\theoremstyle{definition}
\newtheorem{definition}[theorem]{Definition}
\newtheorem{assumption}[theorem]{Assumption}

\theoremstyle{remark}

\newenvironment{keywords}{\par\vspace{0.5em}\noindent\textbf{Keywords: }}{\par}
\newenvironment{classcode}{\par\noindent\textbf{JEL Classification: }}{\par\vspace{0.5em}}
\newcommand{\appendices}{\appendix}

\newcommand{\E}{\mathbb{E}}

\newcommand{\R}{\mathbb{R}}
\newcommand{\F}{\mathcal{F}}
\newcommand{\G}{\mathcal{G}}
\newcommand{\Hh}{\mathcal{H}}
\newcommand{\LL}{\mathcal{L}}
\newcommand{\one}{\mathbf{1}}
\newcommand{\Cov}{\operatorname{Cov}}

\newcommand{\diag}{\operatorname{diag}}
\newcommand{\rank}{\operatorname{rank}}

\begin{document}

\title{Screening-Off Information and Conditional Risk in Portfolio Choice}

\author{Alejandro Rodr\'iguez Dom\'inguez\thanks{
Corresponding author. Email: arodriguez@miraltabank.com;
ORCID: \href{https://orcid.org/0000-0002-2400-1097}
{0000-0002-2400-1097}}\\[0.35em]
\small Quantitative Analysis and AI Department, Miralta Finance Bank S.A., Madrid, Spain\\
\small Department of Computer Science, University of Reading, Reading, United Kingdom\\
\small Department of Data and AI, Albert School, Paris, France}

\date{}

\maketitle

\begin{abstract}
Conditional portfolio models estimate risk relative to a chosen information
set, yet rarely test whether that information removes common cross-asset
dependence. When it does not, systematic risk may be treated as idiosyncratic,
distorting portfolios and attainable efficient frontiers. We formulate this
prior problem as screening-off for portfolio choice. A hierarchy separates
causal, distributional and second-moment requirements, while an exact
covariance decomposition distinguishes represented systematic response,
omitted response risk and residual cross-asset dependence. The two
representation errors have different financial implications: omitted response
necessarily overstates attainable mean--variance opportunities, whereas
incomplete screening can distort them in either direction. Exact finite
perturbation identities and non-asymptotic bounds map both errors into changes
in constrained portfolio weights and frontier potential. Controlled experiments
verify the mechanisms under static, heavy-tailed, dynamic and nonlinear
designs. Frozen out-of-sample market tests combine a broad universe of
observable economic and financial drivers with 150 U.S. equities and 17
hedge-fund strategy indices. Compact selected representations materially reduce
residual dependence in both panels; in equities, residual-aware covariance
estimation improves materially on the uncorrected diagonal-residual
restriction while remaining competitive with established covariance
regularizers. The framework provides a falsifiable information
criterion for conditional risk and a direct map from representation failure
to portfolio instability.
\end{abstract}


\begin{classcode}C58; C61; G11\end{classcode}

\noindent\textbf{Word count: }approximately 10,900 words, including references and appendices.

\section{Introduction}\label{sec:introduction}

Conditional information determines the moments and investment opportunities perceived by a portfolio at the decision date \citep{Merton1973,HansenRichard1987,FersonSiegel2001,PenarandaSentana2016}. Most models nevertheless take the conditioning representation as given. High-dimensional implementations may select variables for prediction or realized portfolio performance \citep{Pereira2021}, but neither criterion establishes that the selected information accounts for the common dependence relevant to conditional risk. The unresolved question is prior to covariance regularization and optimization: when is an observable representation adequate for portfolio choice, and what is the financial consequence when it is not?

This question matters because two representations can deliver similar return forecasts while leaving very different dependence in their residuals. Mean--variance portfolios are already sensitive to moment estimation \citep{Michaud1989,BestGrauer1991}; shrinkage and constraints can stabilize allocations \citep{JagannathanMa2003,LedoitWolf2004}, but they do not determine whether common risk has been represented or merely absorbed statistically. Treating residual dependence as idiosyncratic can understate risk and exaggerate the attainable frontier. At the same time, a representation may be entirely adequate for second-moment portfolio choice without establishing full conditional independence or identifying causal variables. A useful theory must therefore distinguish portfolio sufficiency from stronger probabilistic and structural claims.

We address this gap by formulating screening-off as an information criterion for conditional portfolio risk. Causal separation denotes a common-state representation under explicit structural assumptions; distributional screening-off requires conditional factorization of horizon returns; and second-moment separation requires only diagonal conditional residual covariance. The hierarchy connects common-cause and graphical formulations \citep{Reichenbach1956,Pearl2009} to the weaker object required by mean--variance choice \citep{Markowitz1952}. It also separates decision-time information, which may determine portfolio weights, from horizon-closed information, which is used retrospectively to diagnose common innovations. Observational screening identifies an operative representation, not unique causal labels \citep{PetersJanzingScholkopf2017}.

The paper makes four contributions. First, it gives precise causal, distributional and second-moment definitions and states the assumptions under which the implications pass from the structural level to portfolio covariance. Second, it derives an exact conditional-covariance decomposition into diagonal residual risk, represented systematic response, omitted systematic response and residual cross-asset dependence. The last two terms are economically distinct: omitted response risk is positive semidefinite, whereas residual screening error has no fixed sign. Third, it derives the constrained Markowitz solution on the admissible portfolio subspace together with exact finite perturbation identities and non-asymptotic bounds for weights and frontier potential. Fourth, it turns the theory into an auditable diagnostic and covariance-estimation procedure with frozen out-of-sample evaluation. This contribution is distinct from a factor-selection score: it identifies what a representation must remove, what remains when it fails and how that failure reaches the portfolio decision.

The evidence evaluates those claims rather than proposing a new performance ranking. Controlled experiments verify the information-timing distinction, the two representation-error channels, the limits of observational equivalence and the failure of an affine risk representation under state-varying loadings. Frozen market tests then show that compact same-calendar representations reduce residual dependence materially in a 150-stock U.S. equity panel and in 17 BarclayHedge hedge-fund strategy indices. In the equity panel, residual regularization improves materially on the uncorrected diagonal-residual estimator and remains competitive with standard covariance benchmarks. The empirical claims are therefore confined to second-moment representation and its portfolio consequences.

Section~\ref{sec:literature} positions the framework within conditional portfolio choice, covariance estimation and causal screening. Section~\ref{sec:theory} develops the theory. Section~\ref{sec:numerical} combines controlled experiments and market-data evidence, and Section~\ref{sec:conclusion} concludes.

\section{Related literature}\label{sec:literature}

The literature can be read as a sequence of increasingly explicit choices about information and risk. In dynamic asset-pricing models, conditioning variables determine investment opportunities, conditional moments and pricing restrictions \citep{Merton1973,HansenRichard1987,HansenJagannathan1991}. Portfolio research then brought conditioning information directly into portfolio construction. \citet{FersonSiegel2001} derive minimum-variance efficient portfolios that exploit conditioning information, \citet{FersonSiegelXu2006} characterize mimicking portfolios built from predetermined state variables, \citet{Chiang2015} studies benchmark-relative portfolio choice when managers possess conditioning information, and \citet{PenarandaSentana2016} analyze mean--variance frontiers with conditioning information. This literature establishes why conditioning matters for portfolio choice, but the conditioning representation is typically supplied as part of the model or as a set of predictor variables.

As the number of candidate variables grows, the problem becomes one of statistical information selection. \citet{Pereira2021} studies high-dimensional conditioning information using variable selection, shrinkage, principal-component regression and partial least squares evaluated by out-of-sample portfolio criteria. More broadly, the sensitivity of mean--variance portfolios to estimated moments is well known \citep{Michaud1989,BestGrauer1991,KanZhou2007}, and parameter or model uncertainty has motivated robust and multi-prior portfolio methods \citep{GarlappiUppalWang2007}. These approaches ask which variables or models improve prediction, estimation or portfolio performance. The criterion developed here is different: the selected information is assessed by the cross-asset dependence that remains after conditioning. Predictive relevance and screening are therefore related but distinct properties.

A parallel literature addresses the risk-estimation problem once a representation has been specified. Classical and high-dimensional factor models reduce covariance dimension \citep{ChamberlainRothschild1983,FanFanLv2008}, while characteristic-based latent-factor models provide another conditional representation of the cross-section \citep{KellyPruittSu2019}. Covariance shrinkage stabilizes estimation without requiring an explicit economic factor representation \citep{LedoitWolf2004,LedoitWolf2022}; approximate-factor estimators such as POET allow sparse residual dependence after extracting latent components \citep{FanLiaoMincheva2013}. Portfolio constraints can themselves regularize weights \citep{JagannathanMa2003}, and improvements in covariance estimation need not translate mechanically into portfolio gains \citep{DeMiguelGarlappiUppal2009}. Recent work extends this line through nonlinear covariance filtering \citep{BongiornoChallet2024}, portfolio-specific evaluation of covariance estimators \citep{DomHowardJansenLohre2025}, mixed-frequency factor information \citep{LuWang2025}, and dynamic hierarchical covariance estimation \citep{WangPapailiasVentre2026}. These methods regularize or estimate a covariance object conditional on a chosen representation. Screening addresses an earlier question: what observable information should be conditioned on, and what covariance structure follows when that information removes common dependence?

This distinction also clarifies the relation to factor models. Second-moment separation together with affine response to observed innovations produces a diagonal-plus-low-rank covariance of familiar factor form. Low rank, however, is not implied by screening itself; it follows from the additional finite-dimensional response representation. The observed innovation coordinates also need not coincide with principal components. The matched-dimension principal-component comparisons in the numerical analysis therefore serve as a natural statistical benchmark rather than as a competing definition of the information set. In this sense, screening complements rather than replaces factor modelling or shrinkage: it specifies an information-conditioned decomposition on which those estimation methods can subsequently operate.

The final connection is to the probabilistic and causal literature on screening-off. The common-cause idea originates with \citet{Reichenbach1956} and is formalized through conditional-independence structures in graphical models \citep{Pearl1988,Pearl2009}. Modern causal inference distinguishes observational conditional independence from invariance under intervention \citep{PetersBuhlmannMeinshausen2016,PetersJanzingScholkopf2017}, while general conditional-independence testing is known to be statistically difficult \citep{ShahPeters2020}. Earlier work introduces common-driver sensitivity portfolios and a broader causal-portfolio framework \citep{RD2023MLWA,RD2025Causal}. The contribution here is not the general notion of screening-off, which has a much longer history, but its portfolio-specific combination with a common-cause structural representation and the conditional-risk consequences derived from that separation. We call this construction \emph{causal separation for portfolio choice}. Screening-off defines the probabilistic separation criterion; the common-cause model supplies the causal interpretation. Observationally equivalent separators can therefore generate the same conditional portfolio objects without identifying the same causal variables.

\section{Theoretical developments}\label{sec:theory}

\subsection{Portfolio information and separation}\label{sec:screening}

\subsubsection{Decision information and horizon-closed information}

Let $\mathcal A=\{1,\ldots,n\}$ be the asset universe and let $X=(X^1,\ldots,X^K)$ be a finite declared universe of observable candidate drivers. Fix a decision date $t$, lookback $\delta>0$, and horizon $h>0$. Let $r_{t+h}=(r^1_{t+h},\ldots,r^n_{t+h})^{\top}$ denote the vector of asset returns over the horizon $(t,t+h]$. If $R^i$ denotes the observed history of asset $i$, define the decision-time asset information by
\begin{equation}
\F^{\mathcal A}_t
:=
\bigvee_{i=1}^n\sigma(R^i_s:s\in[t-\delta,t]).
\label{eq:asset-info}
\end{equation}
Thus $\F^{\mathcal A}_t$ contains coordinatewise asset histories, rather than an arbitrary sigma-algebra of market aggregates. For $\mathcal D\subseteq\{1,\ldots,K\}$ define
\begin{equation}
\F^{\mathcal D}_{t,0}:=\sigma(X^j_s:j\in\mathcal D,\ s\in[t-\delta,t]),
\qquad
\F^{\mathcal D}_{t,h}:=\sigma(X^j_s:j\in\mathcal D,\ s\in[t-\delta,t+h]).
\label{eq:driver-info}
\end{equation}
The two nested information sets are
\begin{equation}
\G_t(\mathcal D):=\F^{\mathcal A}_t\vee\F^{\mathcal D}_{t,0},
\qquad
\Hh_{t,h}(\mathcal D):=\F^{\mathcal A}_t\vee\F^{\mathcal D}_{t,h}.
\label{eq:two-information-sets}
\end{equation}
Portfolio weights $w_t$ are $\G_t(\mathcal D)$-measurable. The larger $\Hh_{t,h}(\mathcal D)$ is a conditioning object used to express screening after the common horizon innovations have realized; it is not information available when the portfolio is chosen.

\begin{definition}[Exact distributional separator]\label{def:separator}
A candidate set $\mathcal D$ is an \emph{exact distributional separator} at $(t,h)$ if
\begin{equation}
\LL(r_{t+h}\mid\Hh_{t,h}(\mathcal D))
=
\bigotimes_{i=1}^{n}\LL(r^i_{t+h}\mid\Hh_{t,h}(\mathcal D))
\quad\text{a.s.}
\label{eq:portfolio-separation}
\end{equation}
\end{definition}

\begin{definition}[Second-moment portfolio separator]\label{def:second-moment-separator}
Assume $r_{t+h}\in L^2$. A candidate set $\mathcal D$ is a \emph{second-moment portfolio separator} at $(t,h)$ if
\begin{equation}
\E\!\left[
\Cov(r_{t+h}\mid\Hh_{t,h}(\mathcal D))
\mid\G_t(\mathcal D)
\right]
\quad\text{is diagonal a.s.}
\label{eq:second-moment-separation}
\end{equation}
Equivalently, every off-diagonal horizon-closed residual covariance has zero conditional average at the decision date.
\end{definition}

Both definitions are pointwise in the decision date and horizon. A representation is \emph{$H$-uniformly separating} at $t$ when the relevant property holds for every $h\in(0,H]$. The theoretical statements below are for a fixed $(t,h)$; an estimate at one sampling horizon does not establish the uniform property.

Exact distributional separation requires mutual, not merely pairwise, conditional independence. Under square integrability it implies second-moment separation, while the converse is false. The second-moment definition is deliberately decision-oriented: cancellation across horizon-closed states can make the averaged residual covariance diagonal even when conditional covariance is not diagonal state by state. This distinction aligns the operational object with the portfolio problem: Markowitz choice uses Definition~\ref{def:second-moment-separator}; the stronger Definition~\ref{def:separator} supports distributional and causal interpretations.

The two information windows cannot generally be collapsed. If
\begin{equation}
r^i_{t+h}=b_i\eta_{t,h}+\varepsilon^i_{t,h},
\label{eq:window-example}
\end{equation}
with a common innovation $\eta_{t,h}$ realized after $t$, then conditioning only on decision information leaves the cross-sectional covariance generated by $\eta_{t,h}$. Historical horizon-closed regression can condition on realized $(\eta_{s,h},r_{s+h})$ pairs for $s<t$ without making $w_t$ depend on future information.

Screening and predictive sufficiency are different requirements. Relative to the declared full universe, a selected set can be called sufficient if replacing the full driver information by the selected information leaves the relevant decision-time and horizon-closed conditional laws unchanged. Because the declared universe is finite, inclusion-minimal sufficient sets exist whenever the sufficient class is nonempty, but named-variable representations need not be unique. Two sets that generate the same decision-time and horizon-closed sigma-algebras, up to completion, are observationally equivalent for all conditional-law functionals used below, up to almost-sure versions.

\subsubsection{Causal interpretation under additional structure}

The mean--variance theory requires Definition~\ref{def:second-moment-separator}, not a causal graph. A causal interpretation is available under stronger assumptions and only in one direction.

\begin{assumption}[Common-cause structural partition]\label{ass:causal}
At the fixed horizon, the variables admit an acyclic modular structural causal model with mutually independent exogenous blocks $U^0,U^1,\ldots,U^n$ such that: (i) a common state path $C_{t,h}=c_h(U^0)$ is represented by a declared set $\mathcal D_C$ on both the decision and horizon-closed windows; (ii) each asset path satisfies $R^i_{[t-\delta,t+h]}=f_{i,h}(C_{t,h},U^i)$; and (iii) all common causal channels among the assets over the horizon are contained in $C_{t,h}$. In particular, the observed pre-$t$ history of asset $i$ is measurable with respect to $\sigma(C_{t,h},U^i)$.
\end{assumption}

\begin{definition}[Causal separator for portfolio choice]\label{def:causal-separator}
Under Assumption~\ref{ass:causal}, the declared representation $\mathcal D_C$ that generates the common-state information is called a \emph{causal separator} for portfolio choice. We refer to the associated construction as \emph{causal separation}.
\end{definition}

\begin{proposition}[Structural sufficiency and the separation hierarchy]\label{prop:causal-soundness}
Under Assumption~\ref{ass:causal}, the representation $\mathcal D_C$ generating exactly the common-state information is an exact distributional separator at $(t,h)$. If the horizon returns are square-integrable, it is also a second-moment portfolio separator. Hence, under the stated structural assumptions,
\begin{equation}
\text{causal separator}
\Longrightarrow
\text{exact distributional separator}
\Longrightarrow
\text{second-moment portfolio separator}.
\label{eq:separator-hierarchy}
\end{equation}
A strictly richer conditioning set requires its own product-law or collider-safety check.
\end{proposition}

The proof is given in Appendix~\ref{app:proofs}. Neither converse is claimed. Assumption~\ref{ass:causal} is a transparent sufficient condition, not an identification theorem: it excludes direct asset-to-asset causal channels and within-horizon feedback unless those mechanisms are represented in the common state. The restriction to information generated by the common state is also substantive because conditioning on additional variables can open dependence paths even when the common-cause state itself screens the assets.

\begin{proposition}[Observational equivalence and intervention]\label{prop:obs-intervention}
Suppose an exact distributional separator $C$ and an observable representation $S$ generate the same decision-time and horizon-closed information, up to completion. Then they generate identical conditional laws, up to almost-sure versions, and hence identical portfolio objects. Observational screening therefore identifies operative information, not unique causal labels.

Under Assumption~\ref{ass:causal}, a modular intervention on a variable outside the ancestral system of the common state and horizon returns leaves the common-state conditional laws invariant, provided that the intervention does not alter the conditioning information itself.
\end{proposition}

An invertible recording $P=\psi(C)$ is therefore indistinguishable from $C$ observationally, even if it is only a measurement of the underlying state. If the measurement mechanism is later severed from $C$, the proxy can lose screening while the common state remains valid. This is the distinction tested in the controlled intervention experiment.

\subsection{Conditional risk under screening-off}\label{sec:risk}

Fix a representation and horizon and write $\G_t:=\G_t(\mathcal D)$ and $\Hh_{t,h}:=\Hh_{t,h}(\mathcal D)$. Assume $r_{t+h}$ is square-integrable. Define
\begin{equation}
\mu_{t,h}:=\E[r_{t+h}\mid\G_t],
\qquad
m_{t,h}:=\E[r_{t+h}\mid\Hh_{t,h}].
\label{eq:means}
\end{equation}

\subsubsection{Exact decomposition}

The conditional law of total covariance gives, without a separation assumption,
\begin{equation}
Q_{t,h}:=\Cov(r_{t+h}\mid\G_t)
=
V_{t,h}+\Cov(m_{t,h}\mid\G_t),
\qquad
V_{t,h}:=\E[\Cov(r_{t+h}\mid\Hh_{t,h})\mid\G_t].
\label{eq:general-total-cov}
\end{equation}
Split the first term into its diagonal and off-diagonal parts,
\begin{equation}
D_{t,h}:=\diag(V_{t,h}),
\qquad
E_{t,h}:=V_{t,h}-D_{t,h}.
\label{eq:DE}
\end{equation}
$D_{t,h}$ is diagonal positive semidefinite. The matrix $E_{t,h}$ is symmetric with zero diagonal. If nonzero, it is indefinite: its trace is zero, so it cannot be positive or negative semidefinite unless it vanishes.

\subsubsection{Finite-dimensional innovation representation}
\label{sec:innovation-representation}

Choose a square-integrable horizon innovation vector $\xi_{t,h}\in\R^k$, measurable with respect to $\Hh_{t,h}$, and let
\begin{equation}
\widetilde\xi_{t,h}:=\xi_{t,h}-\E[\xi_{t,h}\mid\G_t],
\qquad
\Lambda_{t,h}:=\Cov(\xi_{t,h}\mid\G_t)\succ0.
\label{eq:innovation}
\end{equation}

Define the conditional linear projection
\begin{equation}
L_{t,h}:=\Cov(m_{t,h},\xi_{t,h}\mid\G_t)\Lambda_{t,h}^{-1},
\qquad
u_{t,h}:=m_{t,h}-\mu_{t,h}-L_{t,h}\widetilde\xi_{t,h}.
\label{eq:loading}
\end{equation}

\begin{theorem}[Exact representation decomposition]\label{thm:risk-decomp}
Under the preceding innovation assumptions, for any candidate representation and square-integrable horizon return,
\begin{equation}
\boxed{
Q_{t,h}
=D_{t,h}+L_{t,h}\Lambda_{t,h}L_{t,h}^{\top}
+\Omega_{t,h}+E_{t,h},
}
\qquad
\Omega_{t,h}:=\Cov(u_{t,h}\mid\G_t)\succeq0.
\label{eq:master}
\end{equation}
Moreover,
\begin{equation}
\Cov(u_{t,h},\widetilde\xi_{t,h}\mid\G_t)=0,
\qquad
\rank(L_{t,h}\Lambda_{t,h}L_{t,h}^{\top})\le k.
\label{eq:projection-properties}
\end{equation}
\end{theorem}

The four terms have distinct meanings. $D$ is the decision-time average of horizon-closed residual variances. $L\Lambda L^{\top}$ is systematic response represented by the chosen innovation span. $\Omega$ is systematic response orthogonal to that span. $E$ is residual cross-asset covariance left by the candidate information. The identity is exact even when the candidate fails to separate the assets and the response is nonlinear.

\begin{corollary}[Restrictions associated with separation and affine response]\label{cor:separation-restrictions}
Under Theorem~\ref{thm:risk-decomp}:
\begin{enumerate}
\item $\mathcal D$ is a second-moment portfolio separator if and only if $E_{t,h}=0$;
\item exact distributional separation implies $E_{t,h}=0$;
\item $\Omega_{t,h}=0$ if and only if $m_{t,h}-\mu_{t,h}$ belongs almost surely to the conditional linear span of $\widetilde\xi_{t,h}$.
\end{enumerate}
Consequently, second-moment separation and affine response jointly imply
\begin{equation}
Q_{t,h}=D_{t,h}+L_{t,h}\Lambda_{t,h}L_{t,h}^{\top},
\qquad
\rank(Q_{t,h}-D_{t,h})\le k.
\label{eq:dlr}
\end{equation}
\end{corollary}

Neither distributional nor second-moment separation implies low rank by itself. With scalar $Z\sim N(0,1)$ and independent idiosyncratic errors, the three conditional means $Z$, $Z^2-1$ and $Z^3-3Z$ are mutually uncorrelated and generate a rank-three systematic covariance even though the separator is one-dimensional. For a nonlinear response satisfying the conditional $L^2$ remainder condition, a small-innovation expansion gives a local low-rank approximation in conditional operator norm rather than an exact identity. Proposition~\ref{prop:local-response-app} in Appendix~\ref{app:proofs} states the precise local result.

\subsection{Decision-time sensitivity geometry}

The innovation loading $L_{t,h}$ describes uncertainty arriving over the horizon. Portfolio sensitivity restrictions use a different object: the response of the decision-time conditional mean to the state already observed at $t$. The decision state introduced below need not coincide with the horizon-innovation coordinates used in Section~\ref{sec:innovation-representation}. Suppose
\begin{equation}
\mu_{t,h}=g_t(Z_t),
\qquad Z_t\in\R^m,
\label{eq:state-map}
\end{equation}
with $g_t$ continuously differentiable locally. Define
\begin{equation}
S_t:=Dg_t(Z_t)\in\R^{n\times m}.
\label{eq:sensitivity}
\end{equation}
If $U_t\in\R^{m\times q}$ spans economically chosen state directions to be neutralized, then
\begin{equation}
w^{\top}S_tU_t=0
\quad\Longleftrightarrow\quad
\Gamma_t w=0,
\qquad \Gamma_t:=U_t^{\top}S_t^{\top}.
\label{eq:constraint}
\end{equation}
The choice of $U_t$ is an admissibility decision, not an output of screening. Under a local diffeomorphism $z'=\psi(z)$, transporting directions as $U'_t=D\psi(Z_t)U_t$ and sensitivities as $S'_t=S_tD\psi(Z_t)^{-1}$ gives $U_t'^{\top}S_t'^{\top}=U_t^{\top}S_t^{\top}=\Gamma_t$. Hence the portfolio restriction depends on the economic directions rather than on the chosen coordinates. Proposition~\ref{prop:geometry-invariance-app} in Appendix~\ref{app:proofs} records the invariance formally.

\subsection{Projected conditional mean--variance choice}\label{sec:portfolio}

Once $(\mu,Q)$ and the admissibility geometry have been specified, the optimization is the classical mean--variance problem \citep{Markowitz1952}. At a fixed decision node write $\Gamma:=\Gamma_t$. For risk tolerance $\gamma\ge0$, stack the budget and sensitivity restrictions as
\begin{equation}
A:=\begin{pmatrix}\one^{\top}\\ \Gamma\end{pmatrix},
\qquad
b:=(1,0,\ldots,0)^{\top},
\label{eq:A}
\end{equation}
with $A$ full row rank, and consider
\begin{equation}
\max_w\left\{\gamma w^{\top}\mu-\frac12w^{\top}Qw\right\}
\quad\text{subject to}\quad Aw=b,
\qquad Q\succ0.
\label{eq:mv}
\end{equation}

\begin{theorem}[Projected conditional Markowitz portfolio]\label{thm:markowitz}
The unique solution of \eqref{eq:mv} is
\begin{equation}
w^{\star}=w_0+\gamma M\mu,
\label{eq:wstar}
\end{equation}
where
\begin{align}
w_0&:=Q^{-1}A^{\top}(AQ^{-1}A^{\top})^{-1}b,\label{eq:w0}\\
M&:=Q^{-1}-Q^{-1}A^{\top}(AQ^{-1}A^{\top})^{-1}AQ^{-1}.\label{eq:M}
\end{align}
If the columns of $N$ form a basis of $\ker A$, then
\begin{equation}
M=N(N^{\top}QN)^{-1}N^{\top},
\label{eq:M-tangent}
\end{equation}
so $M$ is the inverse risk operator restricted to feasible directions.
\end{theorem}

Define
\begin{equation}
\mu_0:=w_0^{\top}\mu,
\qquad
\sigma_0^2:=w_0^{\top}Qw_0,
\qquad
\Delta:=\mu^{\top}M\mu\ge0.
\label{eq:frontier-defs}
\end{equation}
Then
\begin{equation}
\mu_p(\gamma)=\mu_0+\gamma\Delta,
\qquad
\sigma_p^2(\gamma)=\sigma_0^2+\gamma^2\Delta.
\label{eq:frontier}
\end{equation}
If $\Delta>0$, eliminating $\gamma$ gives
\begin{equation}
\sigma_p^2=\sigma_0^2+\frac{(\mu_p-\mu_0)^2}{\Delta}.
\label{eq:frontier-parabola}
\end{equation}
Moreover,
\begin{equation}
\sqrt{\Delta}
=
\sup_{v\in\ker A,\,v\ne0}
\frac{|v^{\top}\mu|}{\sqrt{v^{\top}Qv}},
\label{eq:IR}
\end{equation}
so $\sqrt{\Delta}$ is the maximum conditional information ratio of a zero-cost feasible perturbation. It is not, in general, the maximum Sharpe ratio among fully invested portfolios.

Under second-moment separation, $E=0$ and hence $Q\succeq D$. If, in addition,
\begin{equation}
D\succeq\underline\sigma^2 I,
\label{eq:floor}
\end{equation}
then $Q\succeq\underline\sigma^2I$ and $\|M\|_2\le\underline\sigma^{-2}$. Here and below, $\|\cdot\|_2$ denotes the Euclidean norm for vectors and the spectral norm for matrices. Consequently,
\begin{equation}
\|w^{\star}(\mu+\delta\mu)-w^{\star}(\mu)\|_2
\le
\frac{\gamma}{\underline\sigma^2}\|\delta\mu\|_2.
\label{eq:mean-bound}
\end{equation}
The positive residual variance floor is an additional quantitative assumption, not a consequence of screening. Proposition~\ref{prop:floor-stability-app} in Appendix~\ref{app:proofs} gives the corresponding operator argument.

When the affine covariance representation \eqref{eq:dlr} is used and the positive residual floor \eqref{eq:floor} holds, $D$ is invertible and Woodbury gives
\begin{equation}
Q^{-1}x
=D^{-1}x-D^{-1}L(\Lambda^{-1}+L^{\top}D^{-1}L)^{-1}L^{\top}D^{-1}x.
\label{eq:woodbury}
\end{equation}
For fixed driver and constraint dimensions this reduces the dominant computational cost to linear order in the number of assets; Proposition~\ref{prop:structured-solver-app} in Appendix~\ref{app:proofs} gives the corresponding arithmetic and memory bounds. If $\Omega+E\ne0$ and no additional structure is available, the formula is an approximation rather than the inverse of the true covariance.

\subsection{Misspecification and robustness}\label{sec:robustness}

Second-moment separation and affine response are distinct benchmarks.
Statistical implementations can fail because the chosen innovation span
misses systematic response variation, because the selected information leaves
residual cross-asset covariance, or both.

\subsubsection{Representation-error channels}

Relative to the structured benchmark, define
\begin{equation}
R:=\Omega+E,
\qquad
Q^0:=D+L\Lambda L^{\top},
\qquad
Q=Q^0+R.
\label{eq:R}
\end{equation}
A nonzero $\Omega\succeq0$ is compatible with exact distributional
screening; it says that the chosen innovation span is incomplete or the
response is not affine in those coordinates. A nonzero $E$ means that
second-moment portfolio separation fails. Conversely, $E=0$ establishes only
the second-moment restriction and does not imply mutual conditional
independence or causal identification. This is the theoretical reason to
estimate a residual correction rather than silently impose $R=0$.

\subsubsection{Exact finite covariance perturbations}

For this subsection, let $\bar Q\succ0$ denote any baseline covariance and
let $\bar Q_R:=\bar Q+R\succ0$, holding $(\mu,A,b,\gamma)$ fixed. Let
$(M,w^{\star},\Delta)$ be the projected objects under $\bar Q$ and
$(M_R,w_R^{\star},\Delta_R)$ their counterparts under $\bar Q_R$. This
generic perturbation notation is distinct from the decomposition
$Q=Q^0+R$ in \eqref{eq:R}.

\begin{theorem}[Finite covariance-perturbation identities]
\label{thm:perturbation}
Whenever both covariance matrices are positive definite,
\begin{align}
M_R-M
&=-M R M_R=-M_R R M,
\label{eq:M-resolvent}\\
w_R^{\star}-w^{\star}
&=-M_R R w^{\star}=-M R w_R^{\star},
\label{eq:w-pert}\\
\Delta_R-\Delta
&=-(M\mu)^{\top}R(M_R\mu).
\label{eq:delta-pert}
\end{align}
\end{theorem}

These are finite identities, not first-order approximations. For
$\bar Q_\tau=\bar Q+\tau R$, with $\tau$ in a neighbourhood of zero for
which $\bar Q_\tau\succ0$, they imply
\begin{align}
w_\tau^{\star}-w^{\star}
&=-\tau M R w^{\star}+O(\tau^2),
\label{eq:w-local}\\
\Delta_\tau-\Delta
&=-\tau(M\mu)^{\top}R(M\mu)+O(\tau^2).
\label{eq:delta-local}
\end{align}
If the baseline is the structured covariance
$\bar Q=Q^0=D+L\Lambda L^{\top}$, the floor \eqref{eq:floor} holds, and
$\|R\|_2<\underline\sigma^2$, then
\begin{equation}
\|w_R^{\star}-w^{\star}\|_2
\le
\frac{\|R\|_2}
{\underline\sigma^2-\|R\|_2}
\|w^{\star}\|_2,
\label{eq:w-bound}
\end{equation}
and
\begin{equation}
|\Delta_R-\Delta|
\le
\frac{\|R\|_2}
{\underline\sigma^2(\underline\sigma^2-\|R\|_2)}
\|\mu\|_2^2.
\label{eq:delta-bound}
\end{equation}

The sign structure also matters. If
$Q_{\mathrm{true}}=Q^0+\Omega$ with $\Omega\succeq0$, and
$(M_0,\Delta_0)$ and $(M_{\mathrm{true}},\Delta_{\mathrm{true}})$ are
computed under $Q^0$ and $Q_{\mathrm{true}}$, respectively, then
\begin{equation}
M_{\mathrm{true}}\preceq M_0,
\qquad
\Delta_{\mathrm{true}}\le\Delta_0.
\label{eq:omega-order}
\end{equation}
Omitting systematic response covariance can therefore only overstate feasible
frontier potential. No analogous sign statement holds for the indefinite
residual-screening error $E$.

\subsubsection{Residual covariance diagnostic}

Mean--variance choice depends directly on second moments, so a statistical
implementation can use a weaker second-moment diagnostic than the full
product-law definition. Assume $V_{ii}>0$, put $s_i^2:=V_{ii}$ and define
\begin{equation}
\epsilon^{\mathrm{cov}}(\mathcal D)
:=
\max_{i<j}\frac{|E_{ij}|}{s_is_j}.
\label{eq:eps}
\end{equation}
By Definition~\ref{def:second-moment-separator},
$\epsilon^{\mathrm{cov}}=0$ if and only if the population second-moment
separation restriction holds. Its vanishing does not imply distributional
or causal separation.

\begin{proposition}[Diagnostic-to-operator bound]\label{prop:eps-bound}
If
$|E_{ij}|\le\epsilon^{\mathrm{cov}}s_is_j$
for all $i\ne j$, then
\begin{equation}
\|E\|_2
\le
\|E\|_F
\le
\epsilon^{\mathrm{cov}}
\sqrt{
\left(\sum_i s_i^2\right)^2-\sum_i s_i^4
}
\le
\epsilon^{\mathrm{cov}}\|s\|_2^2.
\label{eq:eps-bound}
\end{equation}
\end{proposition}

Combining \eqref{eq:eps-bound} with the perturbation bounds connects a
pairwise residual-correlation diagnostic to portfolio sensitivity. It does
not transform the diagnostic into a test of mutual conditional independence.
A training diagnostic can be used for representation selection; on a
disjoint segment we distinguish a \emph{set-refit} diagnostic, which fixes
the selected labels and re-estimates nuisance models, from a
\emph{frozen-model} diagnostic, which carries the fitted parameters forward.
Neither is automatically a statistical upper confidence bound. The term
\emph{certified} is reserved for a procedure with an explicit population
coverage guarantee. Corollary~\ref{cor:joint-moment-app} in
Appendix~\ref{app:proofs} gives the simultaneous first-order effect of mean
and covariance perturbations.

\subsubsection{Estimation implication}

The decomposition implies that the uncorrected structured estimator
\begin{equation}
\widehat Q^{0}
:=
\widehat D+\widehat L\widehat\Lambda\widehat L^{\top}
\label{eq:Q0hat}
\end{equation}
imposes both $\Omega=0$ and $E=0$. A residual-aware estimator instead has
the form
\begin{equation}
\widehat Q^{R}
=
\widehat Q^{0}+\widehat R_{\alpha}.
\label{eq:QRhat}
\end{equation}
Here $\widehat R_{\alpha}$ must be estimated and regularized using
information available at the decision date. One positive-semidefinite
construction starts from a residual-aware covariance estimator
$\widetilde Q\succeq0$ fitted on the same training information and sets
\begin{equation}
\widehat Q^{R}_{\alpha}
:=
(1-\alpha)\widehat Q^{0}+\alpha\widetilde Q,
\qquad
\widehat R_{\alpha}
:=
\alpha(\widetilde Q-\widehat Q^{0}),
\qquad
\alpha\in[0,1].
\label{eq:residual-shrinkage}
\end{equation}
This parameterization preserves positive semidefiniteness and makes the
uncorrected restriction ($\alpha=0$) explicit. It does not identify
$\Omega$ and $E$ separately; that distinction remains a population
interpretation unless additional structure is available. The formula is an
estimation implication of the decomposition, not a claim that a particular
residual correction is optimal. Any application must choose $\alpha$ out of
sample and compare \eqref{eq:QRhat} with both \eqref{eq:Q0hat} and the
unrestricted estimator $\widetilde Q$ to distinguish information content
from generic regularization.

\section{Numerical and market evidence}\label{sec:numerical}

The controlled experiments use simulated data-generating processes for which the population information structure is known. Their role is to separate the theoretical implications, assess finite-sample behaviour and check the numerical implementation; the proofs do not rely on simulation evidence. The baseline uses independent Gaussian AR(1) drivers, while the robustness suite adds a heavy-tailed VAR(2), Student-$t$ dynamic conditional correlation, state-varying loadings, random covariance perturbations and gradual proxy failure. Asset loadings are drawn from centred Gaussian distributions and idiosyncratic daily volatilities lie between $0.8\%$ and $2\%$ unless stated otherwise.

\subsection{Two information windows}

The first experiment isolates the distinction between decision information and horizon-closed screening information. A Gaussian panel with $n=30$ assets and $m=3$ common drivers is residualized first on the current state alone and then on the state together with the realized driver innovation. The maximal absolute off-diagonal residual correlation is $0.83$ under past-only conditioning and $0.025$ under horizon-closed conditioning on a long simulated path. Figure~\ref{fig:e1} displays the two residual-correlation matrices.

\begin{figure}[tbp]
\centering
\includegraphics[width=0.72\textwidth]{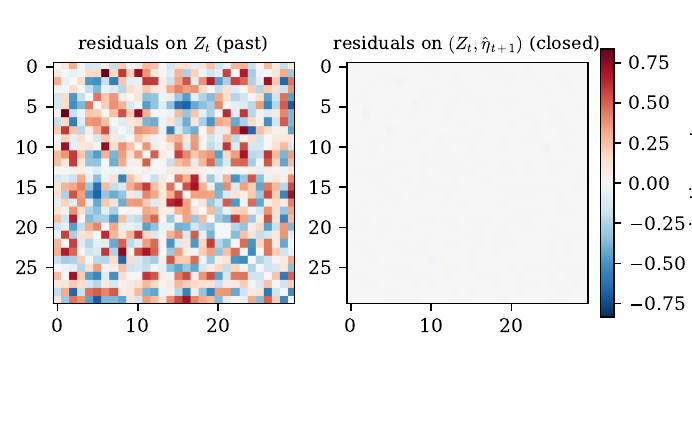}
\caption{Controlled two-window illustration. Conditioning only on the decision-time state leaves the common horizon innovation in the residual covariance (left). Conditioning on the horizon-closed state and realized innovation removes that systematic block in the Gaussian normal form (right). The portfolio decision itself uses only date-$t$ information.}
\label{fig:e1}
\end{figure}

The experiment illustrates the covariance implication of screening rather than estimating the full product law from data. In this Gaussian design the stronger independence statement is known from the construction.

\subsection{Out-of-sample covariance estimation}

The second experiment asks whether exploiting an observed affine innovation structure helps when that structure is correctly specified. Global-minimum-variance portfolios are formed from the sample covariance, Ledoit--Wolf shrinkage, PCA with different factor counts, the structured estimator $\widehat D+\widehat L\widehat\Lambda\widehat L^{\top}$, and the oracle covariance. In each training sample, driver AR dynamics are fitted with intercepts, historical innovations are extracted, returns are regressed on lagged driver states and those innovations, and the covariance is assembled from the estimated innovation covariance and residual variances. The design uses $n\in\{50,100,200\}$, training lengths $T\in\{126,252,504\}$, 250 Monte Carlo replications per cell and 252 out-of-sample observations. Figure~\ref{fig:e2} reports the $T=252$ slice and its variance calibration.

\begin{figure}[tbp]
\centering
\includegraphics[width=0.95\textwidth]{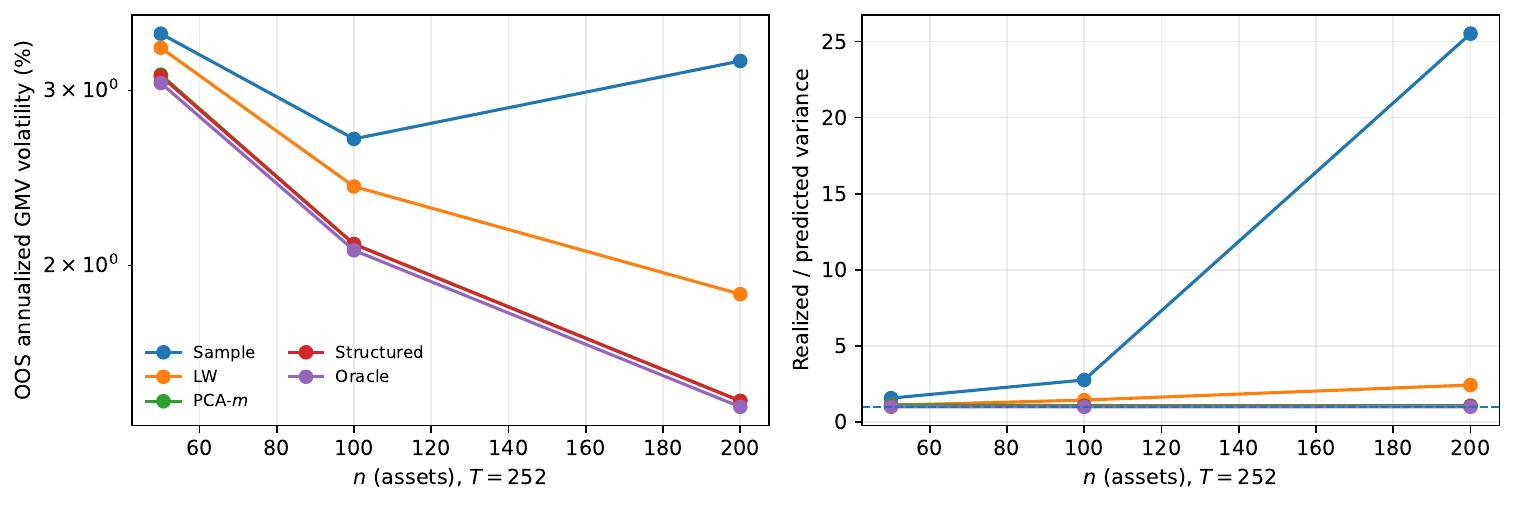}
\caption{Controlled covariance estimation at $T=252$. Left: out-of-sample annualized GMV volatility. Right: realized-to-predicted variance, where one denotes perfect calibration. The structured estimator is close to the oracle under the correctly specified affine design; PCA with the correct factor count provides an important latent-factor benchmark.}
\label{fig:e2}
\end{figure}

The unrestricted sample estimator deteriorates as $n/T$ rises. The structured estimator remains close to the oracle in realized volatility and variance calibration, while PCA supplied with the correct factor dimension is also close. The comparison is therefore not evidence that observable-driver structure must beat a correctly dimensioned latent factor model. It shows that, when the affine screening representation is correct, the theoretical covariance form can be estimated without first recovering latent factors. Ledoit--Wolf remains a strong benchmark and is not uniformly dominated. The full experiment grid is reported in Appendix~\ref{app:experiments}.

\subsection{Separating response error from screening error}

The robustness experiment uses the decomposition $R=\Omega+E$ directly. The baseline has $n=60$ assets and $m=3$ Gaussian driver innovations. Let $\underline\sigma^2:=\min_iD_{ii}$ and parameterize perturbation size by
\begin{equation}
\rho:=\frac{\|R\|_2}{\underline\sigma^2}\in[0,0.8].
\label{eq:rho}
\end{equation}
For the response channel, screening is exact and the horizon-closed conditional mean contains a quadratic Hermite component orthogonal to the linear innovation span, producing $\Omega_\rho\succeq0$ and $E=0$. For the screening channel, the systematic response is affine and a zero-diagonal residual covariance perturbation produces $E\ne0$ and $\Omega=0$. Figure~\ref{fig:e5} compares the exact portfolio displacement, its first-order approximation and the finite bound across the two channels.

\begin{figure}[tbp]
\centering
\includegraphics[width=0.98\textwidth]{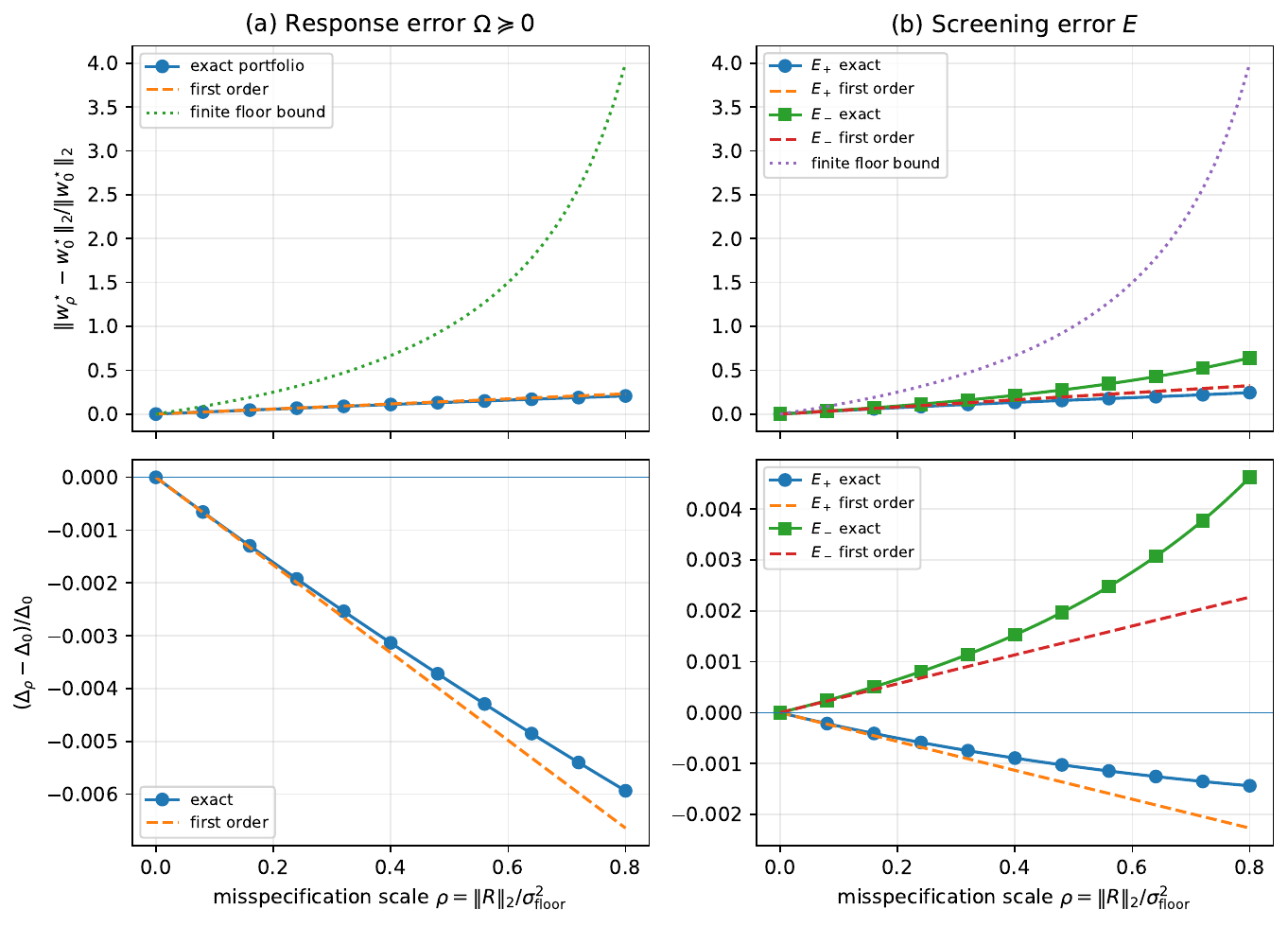}
\caption{Two covariance-misspecification channels. Left: exact screening with nonlinear response and $\Omega_\rho\succeq0$. Right: affine response with zero-diagonal residual screening error $E_\rho$ of either sign. Top panels show relative portfolio displacement with the first-order approximation and finite bound; bottom panels show the relative change in frontier potential.}
\label{fig:e5}
\end{figure}

Across the perturbation grid, the exact identities in Theorem~\ref{thm:perturbation} hold to numerical precision. At $\rho=0.5$, omitted response risk moves the portfolio by $13.4\%$ in relative Euclidean norm and reduces $\Delta$ by $0.39\%$; the first-order values are $14.4\%$ and $-0.42\%$. Reversing the sign of the residual-screening perturbation reverses the direction of its effect on $\Delta$, illustrating why the one-sided result is specific to $\Omega\succeq0$.

\subsection{Observational equivalence broken by intervention}

The final controlled experiment distinguishes observational equivalence from causal invariance. Returns are generated from a two-state normal form with $n=40$. The common-state representation uses $(Z^1,Z^2)$. The alternative records the first state through
\begin{equation}
P_t=-0.4+1.7Z^1_t,
\label{eq:proxy}
\end{equation}
so $\sigma(P,Z^2)=\sigma(Z^1,Z^2)$ before intervention. With intercepts in both stages of estimation, the two representations produce fitted covariances, conditional means and portfolio weights that agree to numerical precision across 250 Monte Carlo replications. At the intervention date the proxy measurement equation is replaced by
\begin{equation}
do\!\left(P_t:=-0.4+1.7\widetilde Z^1_t\right),
\label{eq:proxy-intervention}
\end{equation}
where $\widetilde Z^1$ is an independent stationary AR(1) copy with the same marginal process law as $Z^1$. The structural equations for the true state and returns are unchanged. Figure~\ref{fig:e6} shows that the two representations coincide before the intervention and separate afterwards.

\begin{figure}[tbp]
\centering
\includegraphics[width=0.95\textwidth]{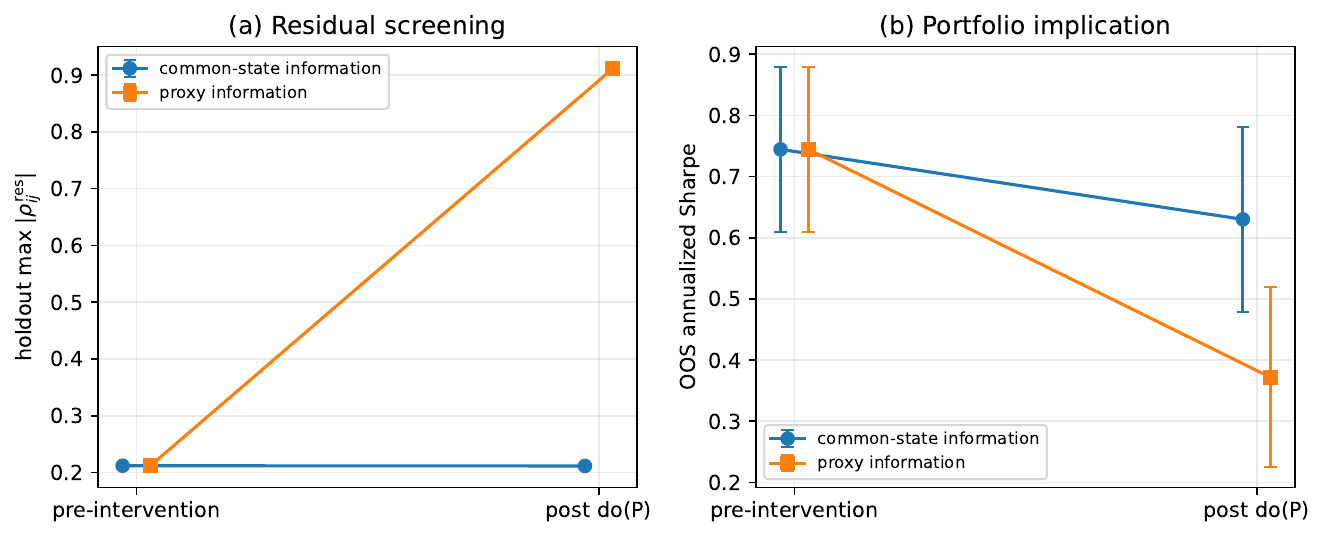}
\caption{Information-equivalent proxy under intervention. Points are Monte Carlo means and error bars are $\pm2$ Monte Carlo standard errors. Before intervention the common-state and proxy representations generate numerically identical portfolio objects. After the proxy measurement mechanism is severed from the state, residual screening deteriorates for the proxy while the common-state representation remains stable.}
\label{fig:e6}
\end{figure}

Before intervention the mean holdout maximal absolute residual correlation is $0.212$ for both representations. After intervention it is $0.211$ for the common-state model and $0.911$ for the proxy. Frictionless synthetic Sharpe ratios are identical before intervention and separate afterwards. These return statistics are model diagnostics, not claims about implementable performance. The identification result is the divergence in screening after the intervention, not the pre-intervention fit.

\subsection{Dynamic misspecification and graded robustness}

The additional covariance experiment removes reliance on the AR(1)-Gaussian benchmark. Four designs are considered: Gaussian AR(1), heavy-tailed VAR(2), Student-$t$ dynamic conditional correlation in the sense of \citet{Engle2002}, and a nonlinear state-loading design. Each design has 80 independent Monte Carlo replications. Figure~\ref{fig:e9} compares relative Frobenius loss and realized-to-predicted GMV variance. Under the first three designs, the structured estimator's median calibration is 1.032, 1.021 and 1.026, respectively. Under state-varying loadings it falls to 0.749; blending in residual covariance raises it to 0.892. The exercise therefore supports the affine representation when it is correctly specified and exposes, rather than conceals, its failure under nonlinear loadings.

\begin{figure}[tbp]
\centering
\includegraphics[width=0.96\textwidth]{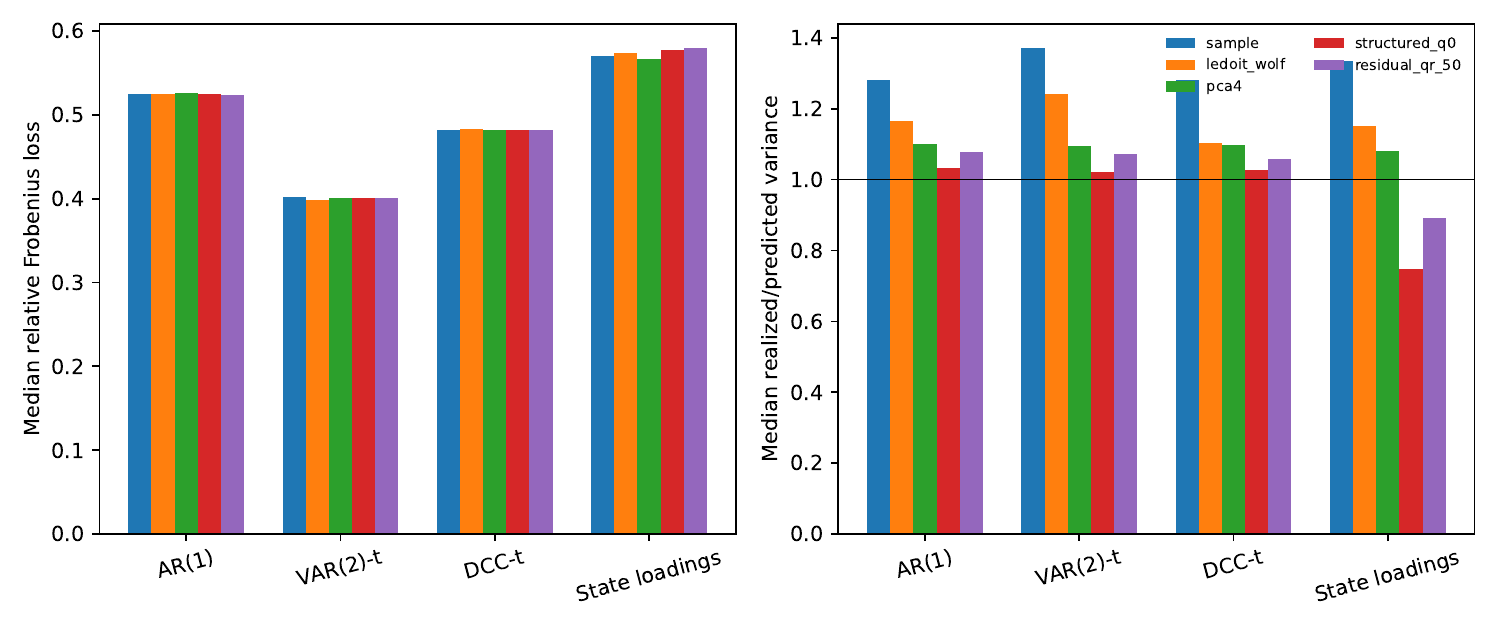}
\caption{Controlled covariance results beyond the AR(1)-Gaussian baseline. Left: median relative Frobenius loss. Right: median realized-to-predicted variance; one denotes calibration. VAR(2)-$t$ and DCC-$t$ preserve the main calibration result, while state-varying loadings create material underprediction that residual regularization partly repairs.}
\label{fig:e9}
\end{figure}

The perturbation identities are also tested over 200 random orientations per channel and perturbation size. The finite bound covers every draw. At $\rho=0.5$, median relative weight displacement is 3.2\% for omitted response, 9.7\% for incomplete screening and 4.4\% when both channels are present; at $\rho=0.8$ the corresponding values are 5.5\%, 15.3\% and 6.5\%. Screening perturbations move the frontier in either direction, whereas the response channel retains the one-sided effect implied by $\Omega\succeq0$.

Finally, the binary proxy intervention is replaced by a grid of gradual breaks and measurement noise. With an exact pre-break proxy the residual-score gap rises monotonically from zero to 0.516 as the measurement link is severed. Noise alone also matters: at noise standard deviation 0.3 the gap is already 0.112 without a structural break and reaches 0.505 under a complete break. Figure~\ref{fig:e10} reports both extensions.

\begin{figure}[tbp]
\centering
\includegraphics[width=0.96\textwidth]{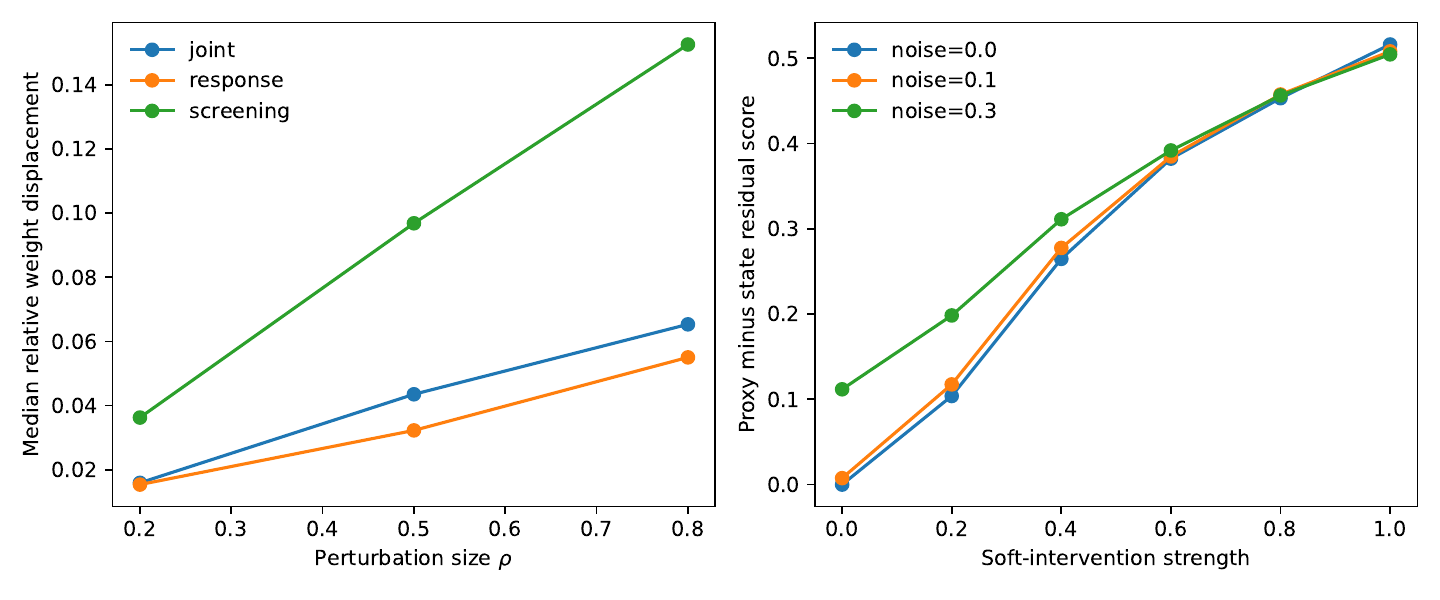}
\caption{Graded robustness. Left: median portfolio displacement under randomly oriented response, screening and joint perturbations. Right: deterioration of a proxy representation as measurement noise and soft-intervention strength increase.}
\label{fig:e10}
\end{figure}

Taken together, the experiments establish five limited but complementary points. Horizon-closed information is needed to diagnose common innovations without introducing look-ahead into the portfolio decision; a correctly specified affine representation can support competitive covariance estimation under more than one driver process; state-varying response is a material failure mode; omitted response and incomplete screening propagate differently to allocations; and observational equivalence does not imply invariance to intervention. Because every experiment is simulated under a known population model, none of these results is presented as market evidence or as empirical causal identification.

\subsection{Market data, timing and interpretation}\label{sec:market-audit}

The empirical analysis combines three data sources. Daily closing prices for
the fixed panel of 150 ticker-labelled U.S. equities were obtained from
Tiingo \citep{Tiingo2026}. Candidate economic and financial driver data were
sourced from Bloomberg \citep{Bloomberg2026}. The declared driver universe
contains 127 economic and financial series spanning developed- and
emerging-market equity indices, sovereign yield curves and swap rates,
foreign-exchange rates, commodity markets, credit-market measures and
market-volatility indicators. Four series are excluded ex ante for frequency
or data-integrity reasons, leaving 123 transformed daily candidates eligible
for the equity experiment. Eligibility is determined separately after
alignment to each empirical panel and therefore need not be identical across
applications.

The hedge-fund panel contains monthly returns for 17 BarclayHedge hedge-fund
strategy indices \citep{BarclayHedge2026}. The equity sample runs from
2 July 2010 to 30 June 2023. The hedge-fund sample runs from January 1997
to December 2025, with 269 complete month-end observations after alignment
with eligible drivers.

Equity prices and positive price-like driver series are converted to log
changes, while yields, rates and other series for which log changes are not
well defined are first-differenced according to a frozen field-level
transformation ledger. Declared missing-value sentinels are removed.
Local-market closures may carry the last driver observation for at most three
source rows before transformation; longer gaps remain missing. Truncated
series and low-frequency releases that cannot be aligned to the relevant
calendar are excluded. The stock panel adjusts documented split-like
discontinuities and excludes reviewed distribution dates without bridging
missing asset observations.

Every empirical specification uses rolling, temporally ordered folds. Model selection, transformations and coefficients are estimated on the training block and then frozen for the non-overlapping test block. The screening statistic, dimension penalty, model-size limit, fold boundaries, random seeds and multiplicity adjustment are fixed. This design makes the market exercises reproducible for researchers with authorized access to the underlying data.

\subsection{Residual-covariance screening}

The screening analysis uses 150 Tiingo U.S. equities with 3,196 complete
daily observations and 123 eligible Bloomberg candidate drivers, and 17
BarclayHedge hedge-fund strategy indices with 269 complete monthly
observations and 119 eligible Bloomberg candidate drivers. The selector
minimizes $s_F(\mathcal D)+0.006|\mathcal D|$, where $s_F$ is normalized
off-diagonal Frobenius residual correlation, subject to
$|\mathcal D|\leq 6$. Equity folds use 504 training sessions followed by
126 test sessions; hedge-fund folds use 120 training months followed by
12 test months. Selection occurs only on the training block and all
coefficients are frozen before evaluating the next test block.

Table~\ref{tab:market-screening} reports the frozen outer-fold results.
Equities improve in all 21 folds, with a 38.1\% median reduction and
Holm-adjusted $p=0.000010$. The hedge-fund indices improve in all 12
folds, with a 25.0\% median reduction and adjusted $p=0.000488$. Median
selected dimensions are three for equities and six for hedge-fund indices.
VIX is selected in every equity fold, while the remaining labels vary
across folds and panels; the result supports compact,
representation-specific screening rather than a universal named factor set.

\begin{table}[tbp]
\centering
\caption{Frozen multi-asset residual-covariance screening. The score change is selected minus unconditioned normalized off-diagonal Frobenius residual correlation; negative values favour screening. Confidence intervals and $p$-values use paired fold resampling and Holm adjustment within each family.}
\label{tab:market-screening}
\small
\begin{tabular}{lrrrrrr}
\hline
Panel & Assets & Obs. & Folds & Median red. & Mean change & Holm $p$ \\
\hline
Equities & 150 & 3,196 & 21 & 38.1\% & $-0.1229$ & 0.000010 \\
BarclayHedge hedge-fund indices & 17 & 269 & 12 & 25.0\% & $-0.1740$ & 0.000488 \\
\hline
\end{tabular}
\end{table}

Two equity robustness checks reinforce the result. State interactions use the training-sample median VIX and add high-state intercept and slope terms without changing the selected labels; they reduce the residual score by a further 0.0090 on average (Holm-adjusted $p=0.000170$). Non-overlapping five-session returns retain a 26.5\% median reduction across 20 folds (adjusted $p=0.000010$). Thus the daily result is not confined to a single-day horizon or a constant loading relation.

\subsection{Covariance estimation and GMV portfolios}

The covariance experiment uses the same fixed panel of 150 U.S. equities
obtained from Tiingo \citep{Tiingo2026}. It contains 3,247 synchronous daily
log price-return dates from 2 July 2010 to 30 June 2023 and 3,191 complete
dates after intersection with the eligible Bloomberg driver information. The
small difference from the 3,196 observations used for screening reflects the
separate driver-alignment step of each exercise.
Documented split-like discontinuities are adjusted, reviewed distribution
dates are excluded, and missing prices are neither imputed nor bridged.

For training windows of 252 and 504 sessions, each followed by a non-overlapping 126-session test block, the analysis compares the sample covariance, linear and quadratic nonlinear shrinkage \citep{LedoitWolf2004,LedoitWolf2022}, OAS, ridge regularization, matched and expanded PCA, POET soft residual thresholding \citep{FanLiaoMincheva2013}, equal weighting, the diagonal-residual estimator $\widehat Q^0$ and the residual-aware estimator $\widehat Q^R_\alpha$ in \eqref{eq:residual-shrinkage}. The mixing weight is selected inside each training window; its median is $\alpha=0.50$ for both designs. The drivers that define the conditioning representation in each covariance fold are not selected on the 150-stock panel itself. Selecting drivers on the same panel whose covariance is being evaluated would let the choice of conditioning information adapt to the evaluation assets, so the experiment instead takes, for each fold, the most recent frozen same-date driver selection certified before the fold's training end from a screening run on a separate 12-equity development panel (102 eligible candidate drivers, 504 selection and 252 certification sessions, dimension penalty 0.01, at most eight drivers). The selected driver sets in the covariance folds contain two to four variables. Driver selection and covariance evaluation therefore use disjoint asset panels, and the covariance results measure the value of a given driver representation rather than of driver search on the evaluation panel. Table~\ref{tab:market-gmv} reports pooled annualized volatility and the median outer-fold ratio of realized to predicted variance for selected benchmarks.

\begin{table}[tbp]
\centering
\caption{Exploratory unconstrained GMV results. Volatility is pooled and annualized; calibration is the median outer-fold realized-to-predicted variance ratio.}
\label{tab:market-gmv}
\small
\begin{tabular}{lrrrr}
\hline
& \multicolumn{2}{c}{252-day training} & \multicolumn{2}{c}{504-day training} \\
Estimator & Vol. (\%) & Calibration & Vol. (\%) & Calibration \\
\hline
Equal weight & 18.74 & 0.95 & 18.46 & 0.91 \\
Sample covariance & 17.12 & 8.80 & 14.48 & 2.65 \\
Ledoit--Wolf & 13.68 & 3.20 & 13.79 & 1.99 \\
Quadratic nonlinear shrinkage & 13.39 & 1.96 & 13.69 & 1.75 \\
POET, soft threshold & 13.21 & 2.71 & 13.52 & 2.62 \\
Ridge, 10\% & 13.43 & 2.18 & 13.55 & 1.56 \\
$\widehat Q^0$ & 15.01 & 6.40 & 15.84 & 5.86 \\
$\widehat Q^R$ & 13.40 & 2.32 & 13.67 & 2.23 \\
\hline
\end{tabular}
\end{table}

Multiplicity-adjusted paired tests show lower realized variance for $\widehat Q^R$ than for equal weighting, matched PCA, PCA plus two and $\widehat Q^0$ in both training designs; it also improves on the sample covariance with 252 training days. The estimator is statistically indistinguishable from ridge, Ledoit--Wolf, OAS, quadratic nonlinear shrinkage and POET. The evidence therefore supports the paper's targeted claim: modelling the remaining residual covariance materially repairs the diagonal-residual restriction while remaining competitive with established regularization.

The covariance result survives two robustness checks. In long-only portfolios, $\widehat Q^R$ volatility is 14.23\% and 14.54\%, with calibration ratios of 1.23 and 1.45. Removing every reviewed corporate-action or distribution date gives volatilities of 13.80\% and 14.06\%. In both cases, $\widehat Q^R$ continues to improve on equal weighting and $\widehat Q^0$ and remains statistically comparable to the strong shrinkage estimators.

\FloatBarrier
\section{Conclusion}\label{sec:conclusion}

Conditional portfolio choice requires more than estimating moments accurately:
it requires establishing whether the information used to condition those
moments represents the common dependence relevant to the portfolio decision.
This paper makes that requirement explicit and testable. It distinguishes
causal, distributional and second-moment separation and shows that
mean--variance portfolio choice requires only the weakest of these conditions.
Causal identification is therefore sufficient for the portfolio problem under
the stated structural assumptions, but it is not necessary for second-moment
portfolio separation.

The resulting covariance decomposition identifies two fundamentally different
forms of representation failure. Omitted systematic response is positive
semidefinite and, when ignored, necessarily overstates attainable
mean--variance opportunities. Residual screening error has no fixed sign and
can distort the frontier in either direction. Exact finite perturbation
identities and non-asymptotic bounds then show how these failures propagate
from the conditioning representation to portfolio weights and frontier
potential. A diagonal-plus-low-rank covariance model is consequently not an
automatic consequence of conditioning on observed drivers: it is a joint
restriction on screening and the response representation that should be
tested before it is imposed.

The empirical evidence supports the operational relevance of this distinction.
Across frozen out-of-sample tests, compact representations selected from a
broad economic and financial information universe reduce residual dependence
in every test fold for both a 150-stock U.S. equity panel and 17 hedge-fund
strategy indices. In equities, explicitly modelling the residual covariance
left after conditioning materially improves on the diagonal-residual estimator
while remaining competitive with established covariance regularizers. The
controlled experiments further show where the framework should fail:
nonlinear state-dependent response breaks the affine representation, and an
observationally equivalent proxy loses its screening property when its link to
the common state is severed.

The broader implication is that the conditioning information should not be
treated as an innocuous input to portfolio optimization. It is part of the
risk model itself. Screening provides a criterion for determining what that
information must remove, the decomposition identifies what remains when it
does not, and the portfolio results quantify the economic consequences of
that failure. This places representation adequacy logically prior to covariance
regularization and portfolio optimization, without requiring stronger causal
claims than the decision problem supports.

Future work can extend the framework to nonlinear response bases, intraday
information sets and portfolio objectives beyond mean--variance choice, and
study how screening representations evolve across market regimes.

\section*{Acknowledgements}
The author thanks Miralta Finance Bank S.A. and his colleagues, as well as J.~D. Opdyke, Miquel Noguer Alonso, Igor Halperin, Daniel Polakow and Mario Bajo Traver, for valuable discussions. 

\section*{Disclosure statement}
The author is employed by Miralta Finance Bank S.A., which provided institutional access to the licensed market data used in this study. The views expressed are those of the author and do not necessarily reflect those of Miralta Finance Bank S.A. The author reports no other competing interests to declare.

\section*{Funding}
No specific funding was received for this research.

\section*{Declaration of generative AI use}
Generative AI (Anthropic's Claude) was used to assist with language editing, LaTeX consistency checks and the preparation of submission materials. The research design, mathematical results, proofs, code and empirical analyses are the author's own, and the author reviewed and takes full responsibility for all content.

\section*{Data and code availability statement}
The complete code, deterministic configurations, tests and derived outputs are publicly available in the reproducibility package at \url{https://doi.org/10.5281/zenodo.21224465}. The market data are third-party series whose redistribution depends on the providers' licences and are therefore not included: daily U.S. equity prices from Tiingo \citep{Tiingo2026}, candidate economic and financial drivers from Bloomberg \citep{Bloomberg2026}, and hedge-fund strategy indices from BarclayHedge \citep{BarclayHedge2026}. Every dataset is fully described in Section~\ref{sec:market-audit}: asset universe, driver identifiers and transformation rules, sample periods, alignment and adjustment procedures. Together with the deposited code, this description is sufficient to replicate all results; a researcher with licensed access to the three sources needs only to place the raw files in the documented location and run the scripts in the stated order.

\appendices

\section{Proofs and technical details}\label{app:proofs}

This appendix records proofs and secondary derivations supporting the results stated in the main text.

\subsection{Proof of causal soundness}

\begin{proof}[Proof of Proposition~\ref{prop:causal-soundness}]
Condition on the realized common path $C_{t,h}$. By
Assumption~\ref{ass:causal}, each asset path is a measurable function of
$(C_{t,h},U^i)$ and the private exogenous blocks
$U^1,\ldots,U^n$ are mutually independent. To make explicit why conditioning on the observed pre-$t$ asset histories
preserves the required product structure, define
\[
\mathcal X_i
:=
\sigma\!\left(R^i_s:s\in[t-\delta,t+h]\right),
\qquad
\mathcal P_i
:=
\sigma\!\left(R^i_s:s\in[t-\delta,t]\right).
\]
Under Assumption~\ref{ass:causal},
\[
\mathcal X_i\subseteq\sigma(C_{t,h},U^i),
\qquad
\mathcal P_i\subseteq\mathcal X_i.
\]
Conditional on $\sigma(C_{t,h})$, the sigma-algebras
$\mathcal X_1,\ldots,\mathcal X_n$ are therefore mutually independent. For bounded measurable functions $\varphi_1,\ldots,\varphi_n$,
coordinatewise conditioning gives
\[
\E\!\left[
\prod_{i=1}^n\varphi_i(r^i_{t+h})
\,\middle|\,
\sigma(C_{t,h})\vee\bigvee_{i=1}^n\mathcal P_i
\right]
=
\prod_{i=1}^n
\E\!\left[
\varphi_i(r^i_{t+h})
\,\middle|\,
\sigma(C_{t,h})\vee\mathcal P_i
\right]
\quad\text{a.s.}
\]
Indeed, conditional on the common state, each factor and its associated
conditioning history are measurable with respect to the corresponding
private block $U^i$, and the private blocks are mutually independent.
Thus conditioning further on the coordinatewise histories does not introduce
cross-coordinate information. Since
\[
\F^{\mathcal A}_t
=
\bigvee_{i=1}^n\mathcal P_i,
\]
the preceding identity implies
\[
\LL(r_{t+h}\mid\F^{\mathcal A}_t\vee\sigma(C_{t,h}))
=
\bigotimes_{i=1}^n
\LL(r^i_{t+h}\mid\F^{\mathcal A}_t\vee\sigma(C_{t,h})),
\]
which is Definition~\ref{def:separator}. Under square integrability, the
product law makes the horizon-closed covariance diagonal almost surely;
taking its decision-time conditional expectation yields
Definition~\ref{def:second-moment-separator}. The same argument does not
apply automatically after adding arbitrary extra conditioning variables,
hence the restriction in the statement.
\end{proof}

\begin{proof}[Proof of Proposition~\ref{prop:obs-intervention}]
First consider observational equivalence. Suppose the exact distributional
separator $C$ and the observable representation $S$ generate the same
decision-time and horizon-closed sigma-algebras, up to completion. Conditional
expectations with respect to equal sigma-algebras are equal almost surely.
Consequently, for every bounded measurable function $\varphi$,
\[
\E[\varphi(r_{t+h})\mid\Hh^C_{t,h}]
=
\E[\varphi(r_{t+h})\mid\Hh^S_{t,h}]
\quad\text{a.s.},
\]
and analogously at the decision-time information level. Hence the two
representations generate the same conditional distributions, up to
almost-sure versions, and therefore the same conditional-law functionals,
conditional moments and portfolio objects appearing in the paper.

Now consider the intervention statement under
Assumption~\ref{ass:causal}. A modular intervention on a variable outside
the ancestral system of the common state and the horizon returns leaves the
structural assignments generating $(C,r)$ unchanged. In the intervened SCM,
the corresponding truncated factorization therefore leaves the joint law of
the exogenous and endogenous variables in that ancestral subsystem unchanged.
Provided that the intervention also leaves the conditioning information
itself unchanged, the conditional law of $r_{t+h}$ given the common-state
information is consequently unchanged.

Thus observational equality follows from equality of the generated
information, whereas intervention invariance requires the additional
structural assumptions stated in the proposition. In particular,
observational equivalence alone does not identify unique causal labels.
\end{proof}

\subsection{Risk decomposition and local nonlinear response}

\begin{proof}[Proof of Theorem~\ref{thm:risk-decomp}]
Since $\G_t\subseteq\Hh_{t,h}$, the conditional law of total covariance gives
\[
\Cov(r_{t+h}\mid\G_t)
=
\E[\Cov(r_{t+h}\mid\Hh_{t,h})\mid\G_t]
+
\Cov(\E[r_{t+h}\mid\Hh_{t,h}]\mid\G_t).
\]
The first term is $V=D+E$ by \eqref{eq:DE}. The coefficient $L$ in
\eqref{eq:loading} is the conditional $L^2$ projection coefficient of
$m-\mu$ on $\widetilde\xi$. Hence
$\Cov(u,\widetilde\xi\mid\G_t)=0$. Writing
$m-\mu=L\widetilde\xi+u$ and taking conditional covariances removes the
cross terms and gives
\[
\Cov(m\mid\G_t)
=
L\Lambda L^{\top}+\Cov(u\mid\G_t).
\]
Combining the two identities proves \eqref{eq:master}. The rank statement
follows from the dimensions of $L$ and $\Lambda$.
\end{proof}

\begin{proof}[Proof of Corollary~\ref{cor:separation-restrictions}]
Definition~\ref{def:second-moment-separator} is equivalent to $V$ being
diagonal, which is equivalent to $E=0$ by \eqref{eq:DE}. Exact
distributional separation makes
$\Cov(r_{t+h}\mid\Hh_{t,h})$ diagonal almost surely, so its
$\G_t$-conditional expectation is diagonal. Finally,
$\E[u\mid\G_t]=0$, and therefore
$\Omega=\E[uu^{\top}\mid\G_t]$ vanishes if and only if $u=0$ almost surely.
Substitution in \eqref{eq:master} and the rank bound in
\eqref{eq:projection-properties} give \eqref{eq:dlr}.
\end{proof}

\begin{proposition}[Local nonlinear response]\label{prop:local-response-app}
Let $\tau\downarrow0$ scale the horizon innovations as
\[
\xi_{t,h}^{(\tau)}=\tau z_{t,h},
\qquad
\E[z_{t,h}\mid\G_t]=0,
\qquad
\Cov(z_{t,h}\mid\G_t)=\Lambda_{t,h}\succ0.
\]
Suppose the horizon-closed conditional mean satisfies
\begin{equation}
m_{t,h}^{(\tau)}
=
\mu_{t,h}^{(\tau)}
+\tau J_{t,h}z_{t,h}
+R_{t,h}^{(\tau)},
\label{eq:local-response-expansion-app}
\end{equation}
where $\E[R_{t,h}^{(\tau)}\mid\G_t]=0$ and
\begin{equation}
\left(
\E[\|R_{t,h}^{(\tau)}\|_2^2\mid\G_t]
\right)^{1/2}
=o(\tau).
\label{eq:local-response-remainder-app}
\end{equation}
Here and below, the conditional little-$o$ statements are understood almost
surely with respect to the decision-time information: after division by the
indicated deterministic power of $\tau$, the corresponding conditional norm
converges to zero almost surely as $\tau\downarrow0$.

Then, in conditional operator norm,
\begin{align}
\Cov(m_{t,h}^{(\tau)}\mid\G_t)
&=
\tau^2J_{t,h}\Lambda_{t,h}J_{t,h}^{\top}
+o(\tau^2),
\label{eq:local-systematic-app}\\
L_{t,h}^{(\tau)}
&=
J_{t,h}+o(1),
\qquad
\Omega_{t,h}^{(\tau)}
=
o(\tau^2).
\label{eq:local-projection-app}
\end{align}
Under second-moment separation,
\begin{equation}
Q_{t,h}^{(\tau)}
=
D_{t,h}^{(\tau)}
+
\tau^2J_{t,h}\Lambda_{t,h}J_{t,h}^{\top}
+
o(\tau^2).
\label{eq:local-response-app}
\end{equation}
\end{proposition}

\begin{proof}
From \eqref{eq:local-response-expansion-app}, the conditional covariance
contains the leading term $\tau^2J\Lambda J^{\top}$, two cross terms involving
$R^{(\tau)}$, and $\Cov(R^{(\tau)}\mid\G_t)$. Conditional
Cauchy--Schwarz and \eqref{eq:local-response-remainder-app} make the cross
terms and the remainder covariance $o(\tau^2)$ in the stated conditional
operator norm, proving \eqref{eq:local-systematic-app}. Moreover,
\[
L^{(\tau)}
=
\Cov(m^{(\tau)},\tau z\mid\G_t)
(\tau^2\Lambda)^{-1}
=
J+o(1).
\]
The conditional linear projection minimizes residual second moment, so its
residual covariance is $o(\tau^2)$, giving
\eqref{eq:local-projection-app}. The final statement follows from
Theorem~\ref{thm:risk-decomp} with $E=0$.
\end{proof}

\begin{proposition}[Coordinate invariance of sensitivity constraints]
\label{prop:geometry-invariance-app}
Let $z'=\psi(z)$ be a $C^1$ local diffeomorphism. Transport the neutral state
directions by
$U'_t=D\psi(Z_t)U_t$ and write the same conditional-mean map in the new
coordinates, with Jacobian $S'_t$. Then
\begin{equation}
S'_t=S_tD\psi(Z_t)^{-1},
\qquad
U_t'^{\top}S_t'^{\top}
=
U_t^{\top}S_t^{\top}
=
\Gamma_t.
\label{eq:geometry-invariance-app}
\end{equation}
Hence the admissible portfolio subspace does not depend on the local
coordinate chart used for the selected decision state.
\end{proposition}

\begin{proof}
The first identity is the chain rule applied to $g_t\circ\psi^{-1}$.
Substitution into the definition of $\Gamma_t$ gives the second identity.
\end{proof}

\subsection{Projected Markowitz derivation}

\begin{proof}[Proof of Theorem~\ref{thm:markowitz}]
The Lagrangian for \eqref{eq:mv} is
\[
\mathcal L(w,\lambda)
=
\gamma w^{\top}\mu
-\frac12w^{\top}Qw
+\lambda^{\top}(b-Aw).
\]
The first-order conditions give
\[
Qw=\gamma\mu-A^{\top}\lambda,
\qquad
Aw=b.
\]
Substituting
$w=\gamma Q^{-1}\mu-Q^{-1}A^{\top}\lambda$
into the constraint gives
\[
\lambda
=
(AQ^{-1}A^{\top})^{-1}
(\gamma AQ^{-1}\mu-b).
\]
Substitution yields
$w^{\star}=w_0+\gamma M\mu$
with \eqref{eq:w0}--\eqref{eq:M}. If the columns of $N$ form a basis of $\ker A$, the feasible
displacement from any particular feasible point is $Nv$. Minimizing risk on
that tangent space gives the reduced inverse
$(N^{\top}QN)^{-1}$ and therefore \eqref{eq:M-tangent}. Positive
definiteness of $Q$ gives uniqueness.

The frontier identities follow from
$Aw_0=b$, $AM=0$ and $MQM=M$. In particular,
$w_0^{\top}QM=0$ and
\[
(w_0+\gamma M\mu)^{\top}Q(w_0+\gamma M\mu)
=
\sigma_0^2+\gamma^2\mu^{\top}M\mu.
\]
Together with
\[
\mu^{\top}(w_0+\gamma M\mu)
=
\mu_0+\gamma\Delta,
\qquad
\Delta=\mu^{\top}M\mu,
\]
this gives the stated frontier relation after eliminating $\gamma$ when
$\Delta>0$. Finally, using
$M=N(N^{\top}QN)^{-1}N^{\top}$ and writing every feasible zero-cost
perturbation as $Nv$ gives
\[
\sup_{v\in\ker A,\;v\ne0}
\frac{|v^{\top}\mu|}
{\sqrt{v^{\top}Qv}}
=
\sqrt{\mu^{\top}M\mu}
=
\sqrt{\Delta},
\]
which is the Rayleigh-quotient representation in \eqref{eq:IR}.
\end{proof}

\begin{proposition}[Residual-risk floor and mean stability]
\label{prop:floor-stability-app}
If the representation is a second-moment portfolio separator and
$D\succeq\underline\sigma^2I$ with $\underline\sigma^2>0$, then
\[
Q\succeq\underline\sigma^2I,
\qquad
\|Q^{-1}\|_2\le\underline\sigma^{-2},
\qquad
\|M\|_2\le\underline\sigma^{-2},
\]
and, holding $(Q,A,b)$ fixed,
\[
\|w^{\star}(\mu+\delta\mu)-w^{\star}(\mu)\|_2
\le
\gamma\underline\sigma^{-2}\|\delta\mu\|_2.
\]
\end{proposition}

\begin{proof}
Second-moment separation gives $E=0$, so \eqref{eq:master} yields
$Q\succeq D$. Also
\[
M=Q^{-1/2}PQ^{-1/2},
\]
where $P$ is the orthogonal projector onto
$\ker(AQ^{-1/2})$, so
$\|M\|_2\le\|Q^{-1}\|_2$. The last inequality follows from
\[
w^{\star}(\mu+\delta\mu)-w^{\star}(\mu)
=
\gamma M\delta\mu.
\]
\end{proof}

\begin{proposition}[Structured solve at fixed dimension]
\label{prop:structured-solver-app}
Suppose the affine representation is exact,
$Q=D+L\Lambda L^{\top}$, with $D$ diagonal positive definite,
$L\in\R^{n\times k}$ and $p$ equality constraints in $Aw=b$. Then the
projected Markowitz solution can be computed without forming a dense
$n\times n$ inverse, with arithmetic cost
\begin{equation}
O\!\left(n(k+p)^2+(k+p)^3\right)
\label{eq:structured-complexity-app}
\end{equation}
and working memory $O(n(k+p))$. The count assumes standard dense arithmetic
for the $k$- and $p$-dimensional blocks; in the intended regime $k$ and $p$
are small relative to $n$.
\end{proposition}

\begin{proof}
Factor the $k\times k$ matrix
$\Lambda^{-1}+L^{\top}D^{-1}L$ once. Woodbury then gives
$Q^{-1}\mu$ and the $p$ columns of $Q^{-1}A^{\top}$. Forming the required
cross-products costs $O(n(k+p)^2)$; the remaining dense factorizations
involve matrices of dimension at most $k+p$. The stated memory bound follows
from storing $L$, $A^{\top}$ and the corresponding transformed columns.
\end{proof}

\subsection{Finite covariance perturbations and diagnostic bound}

\begin{proof}[Proof of Theorem~\ref{thm:perturbation}]
Let the columns of $N$ form a basis of $\ker A$. By \eqref{eq:M-tangent},
\[
M
=
N(N^{\top}\bar QN)^{-1}N^{\top},
\qquad
M_R
=
N\{N^{\top}(\bar Q+R)N\}^{-1}N^{\top}.
\]
The resolvent identity for the two reduced inverse matrices gives
\[
M_R-M
=
-M R M_R
=
-M_R R M.
\]

For the portfolio weights, subtract the two KKT systems. Since
\[
d:=w_R^{\star}-w^{\star}\in\ker A,
\]
the reduced inverse on the feasible tangent space gives
\[
d=-M_R R w^{\star}.
\]
Equivalently, reversing the roles of $\bar Q$ and $\bar Q+R$ gives
\[
d=-M R w_R^{\star}.
\]
Finally,
\[
\Delta_R-\Delta
=
\mu^{\top}(M_R-M)\mu
=
-(M\mu)^{\top}R(M_R\mu).
\]
This proves the exact finite identities. For the local formulas, replace $R$ by $\tau R$. The resolvent identity
implies
\[
M_{\tau R}
=
M-\tau M R M+O(\tau^2),
\]
whenever the perturbation remains inside the positive-definite domain.
Substitution into the exact portfolio and frontier identities gives the
first-order expansions stated in the theorem. For the finite norm bound, the exact identity
\[
d=-M_R R w^{\star}
\]
and the floor condition imply
\[
\|M_R\|_2
\le
\frac{1}{\underline\sigma^2-\|R\|_2}
\]
whenever $\|R\|_2<\underline\sigma^2$, because Weyl's inequality gives
\[
\bar Q+R
\succeq
(\underline\sigma^2-\|R\|_2)I.
\]
Consequently,
\[
\|w_R^{\star}-w^{\star}\|_2
\le
\frac{\|R\|_2}
{\underline\sigma^2-\|R\|_2}
\|w^{\star}\|_2,
\]
which is \eqref{eq:w-bound}. Similarly,
\[
|\Delta_R-\Delta|
\le
\|M\|_2\,
\|R\|_2\,
\|M_R\|_2\,
\|\mu\|_2^2
\le
\frac{\|R\|_2}
{\underline\sigma^2(\underline\sigma^2-\|R\|_2)}
\|\mu\|_2^2,
\]
giving \eqref{eq:delta-bound}. The alternative generic resolvent condition
$\|M R\|_2<1$ yields, from
$d=-M R(w^{\star}+d)$,
\[
\|d\|_2
\le
\frac{\|M R\|_2}{1-\|M R\|_2}
\|w^{\star}\|_2.
\]
The floor condition
$\|R\|_2<\underline\sigma^2$ is a directly checkable sufficient condition
for the positive-definite perturbation bounds used above; it should not be
confused with the generic resolvent condition. Finally, set $\bar Q=Q^0$ and $R=\Omega\succeq0$. On the feasible tangent
space,
\[
N^{\top}(Q^0+\Omega)N
\succeq
N^{\top}Q^0N.
\]
Inversion reverses the Loewner order, hence
\[
M_{\mathrm{true}}\preceq M_0,
\]
and therefore
\[
\Delta_{\mathrm{true}}
=
\mu^{\top}M_{\mathrm{true}}\mu
\le
\mu^{\top}M_0\mu
=
\Delta_0.
\]
This proves \eqref{eq:omega-order}.
\end{proof}

\begin{proof}[Proof of Proposition~\ref{prop:eps-bound}]
Because $E_{ii}=0$,
\[
\|E\|_F^2
=
\sum_{i\ne j}E_{ij}^2
\le
(\epsilon^{\mathrm{cov}})^2
\sum_{i\ne j}s_i^2s_j^2
=
(\epsilon^{\mathrm{cov}})^2
\left\{
\left(\sum_i s_i^2\right)^2-\sum_i s_i^4
\right\}.
\]
Use $\|E\|_2\le\|E\|_F$ and take square roots.
\end{proof}

\begin{corollary}[Joint local mean and covariance perturbation]
\label{cor:joint-moment-app}
Let
$\bar Q_\tau=\bar Q+\tau R$ and
$\mu_\tau=\mu+\tau a$, with $(A,b,\gamma)$ fixed. At $\tau=0$,
\begin{align}
\left.\frac{d w_\tau^{\star}}{d\tau}\right|_{\tau=0}
&=
\gamma M a-M R w^{\star},
\label{eq:joint-weight-app}\\
\left.\frac{d\Delta_\tau}{d\tau}\right|_{\tau=0}
&=
2a^{\top}M\mu-(M\mu)^{\top}R(M\mu).
\label{eq:joint-delta-app}
\end{align}
Thus mean and covariance error enter through distinct first-order channels.
\end{corollary}

\begin{proof}
Apply Theorem~\ref{thm:perturbation} with covariance perturbation $\tau R$.
Its resolvent identity gives
\[
M_\tau
=
M-\tau M R M+O(\tau^2).
\]
The corresponding minimum-risk feasible component satisfies
\[
w_{0,\tau}
=
w_0-\tau M R w_0+O(\tau^2),
\]
while
\[
\gamma M_\tau\mu_\tau
=
\gamma M\mu
+
\tau\{\gamma M a-\gamma M R M\mu\}
+
O(\tau^2).
\]
Since $w^{\star}=w_0+\gamma M\mu$, summing these terms gives
\[
w_\tau^{\star}
=
w^{\star}
+
\tau\{\gamma M a-M R w^{\star}\}
+
O(\tau^2),
\]
which proves \eqref{eq:joint-weight-app}. For the frontier potential,
\[
\Delta_\tau
=
\mu_\tau^{\top}M_\tau\mu_\tau.
\]
Using
$\mu_\tau=\mu+\tau a$ and
$M_\tau=M-\tau M R M+O(\tau^2)$ gives
\[
\Delta_\tau
=
\Delta
+
\tau
\left\{
2a^{\top}M\mu
-
(M\mu)^{\top}R(M\mu)
\right\}
+
O(\tau^2),
\]
which proves \eqref{eq:joint-delta-app}.
\end{proof}

\section{Additional controlled results}\label{app:experiments}

\subsection{Full covariance-estimation grid}

Table~\ref{tab:e2-full} reports the complete covariance experiment rather
than only the $T=252$ slice displayed in Figure~\ref{fig:e2}. It includes
under- and over-specified PCA dimensions to show that much of the gap to the
structured estimator disappears when PCA is supplied with the correct factor
dimension.

\begin{table}[H]
\centering
\caption{Controlled covariance experiment: out-of-sample annualized GMV
volatility and variance calibration. There are 250 Monte Carlo replications
per cell and 252 out-of-sample observations.}
\label{tab:e2-full}
\scriptsize
\resizebox{\textwidth}{!}{\begin{tabular}{@{}llccccccc@{}}
\hline
$n$ & $T$ & Sample & LW & PCA-$(m{-}2)$ & PCA-$m$ & PCA-$(m{+}2)$ & Structured & Oracle\\ \hline
50 & 126 & 3.91 & 3.44 & 4.30 & 3.15 & 3.20 & 3.15 & 3.05\\
50 & 252 & 3.42 & 3.31 & 4.30 & 3.11 & 3.15 & 3.10 & 3.05\\
50 & 504 & 3.24 & 3.22 & 4.33 & 3.11 & 3.15 & 3.10 & 3.08\\
100 & 126 & 4.68 & 2.55 & 3.04 & 2.15 & 2.16 & 2.15 & 2.08\\
100 & 252 & 2.68 & 2.40 & 3.03 & 2.10 & 2.11 & 2.10 & 2.07\\
100 & 504 & 2.31 & 2.26 & 2.99 & 2.09 & 2.09 & 2.08 & 2.07\\
200 & 126 & 276.90 & 1.87 & 2.16 & 1.48 & 1.49 & 1.49 & 1.44\\
200 & 252 & 3.21 & 1.87 & 2.13 & 1.46 & 1.46 & 1.46 & 1.44\\
200 & 504 & 1.85 & 1.71 & 2.08 & 1.45 & 1.45 & 1.45 & 1.44\\
\hline
\multicolumn{9}{@{}l}{\emph{Calibration: realized/predicted variance (1 = perfect)}}\\
50 & 252 & 1.57 & 1.10 & 1.12 & 1.13 & 1.20 & 1.04 & 1.00\\
100 & 252 & 2.76 & 1.44 & 1.12 & 1.08 & 1.12 & 1.04 & 0.99\\
200 & 252 & 25.52 & 2.43 & 1.10 & 1.07 & 1.09 & 1.04 & 1.00\\
\hline
\end{tabular}
}
\end{table}

\subsection{Finite-perturbation calibration}

The response-error design uses
\[
m_\rho
=
\mu+B\eta
+
\sqrt{\frac{\rho\underline\sigma^2}{2}}\,c
\left[
\left(
\frac{\eta_1}{\sqrt{\Lambda_{11}}}
\right)^2-1
\right],
\]
where $c$ is aligned with the baseline admissible speculative direction.
Since the centred squared Gaussian term has variance two and is orthogonal
to the linear innovation span,
\[
\Omega_\rho
=
\rho\underline\sigma^2cc^{\top},
\qquad
E=0.
\]

For the screening-error design, $\Omega=0$ and a two-asset zero-diagonal
residual block is scaled to have spectral norm
$\rho\underline\sigma^2$. Using both signs shows that $E$ has no one-sided
effect on $\Delta$.

\begin{thebibliography}{99}
\bibitem[BarclayHedge(2026)]{BarclayHedge2026}
BarclayHedge, BarclayHedge Hedge Fund Indices. BarclayHedge, 2026.
Available at: \url{https://www.barclayhedge.com/}.
\bibitem[Best and Grauer(1991)]{BestGrauer1991} Best, M.J. and Grauer, R.R., On the sensitivity of mean-variance-efficient portfolios to changes in asset means: Some analytical and computational results. \textit{Rev. Financ. Stud.}, 1991, \textbf{4}, 315--342.
\bibitem[Bloomberg(2026)]{Bloomberg2026}
Bloomberg L.P., Financial market data, 2026.
\bibitem[Bongiorno and Challet(2024)]{BongiornoChallet2024} Bongiorno, C. and Challet, D., Covariance matrix filtering and portfolio optimisation: the average oracle vs non-linear shrinkage and all the variants of DCC-NLS. \textit{Quant. Finance}, 2024, \textbf{24}, 1227--1234.
\bibitem[Chamberlain and Rothschild(1983)]{ChamberlainRothschild1983} Chamberlain, G. and Rothschild, M., Arbitrage, factor structure, and mean-variance analysis on large asset markets. \textit{Econometrica}, 1983, \textbf{51}, 1281--1304.
\bibitem[Chiang(2015)]{Chiang2015} Chiang, I.-H.E., Modern portfolio management with conditioning information. \textit{J. Empir. Finance}, 2015, \textbf{33}, 114--134.
\bibitem[DeMiguel et al.(2009)]{DeMiguelGarlappiUppal2009} DeMiguel, V., Garlappi, L. and Uppal, R., Optimal versus naive diversification: How inefficient is the 1/N portfolio strategy? \textit{Rev. Financ. Stud.}, 2009, \textbf{22}, 1915--1953.
\bibitem[Dom et al.(2025)]{DomHowardJansenLohre2025} Dom, M.S., Howard, C., Jansen, M. and Lohre, H., Beyond GMV: the relevance of covariance matrix estimation for risk-based portfolio construction. \textit{Quant. Finance}, 2025, \textbf{25}, 403--419.
\bibitem[Engle(2002)]{Engle2002} Engle, R., Dynamic conditional correlation: A simple class of multivariate generalized autoregressive conditional heteroskedasticity models. \textit{J. Bus. Econ. Stat.}, 2002, \textbf{20}, 339--350.
\bibitem[Fan et al.(2008)]{FanFanLv2008} Fan, J., Fan, Y. and Lv, J., High dimensional covariance matrix estimation using a factor model. \textit{J. Econometrics}, 2008, \textbf{147}, 186--197.
\bibitem[Fan et al.(2013)]{FanLiaoMincheva2013} Fan, J., Liao, Y. and Mincheva, M., Large covariance estimation by thresholding principal orthogonal complements. \textit{J. R. Stat. Soc. B}, 2013, \textbf{75}, 603--680.
\bibitem[Ferson and Siegel(2001)]{FersonSiegel2001} Ferson, W.E. and Siegel, A.F., The efficient use of conditioning information in portfolios. \textit{J. Finance}, 2001, \textbf{56}, 967--982.
\bibitem[Ferson et al.(2006)]{FersonSiegelXu2006} Ferson, W.E., Siegel, A.F. and Xu, P., Mimicking portfolios with conditioning information. \textit{J. Financ. Quant. Anal.}, 2006, \textbf{41}, 607--635.
\bibitem[Garlappi et al.(2007)]{GarlappiUppalWang2007} Garlappi, L., Uppal, R. and Wang, T., Portfolio selection with parameter and model uncertainty: A multi-prior approach. \textit{Rev. Financ. Stud.}, 2007, \textbf{20}, 41--81.
\bibitem[Hansen and Jagannathan(1991)]{HansenJagannathan1991} Hansen, L.P. and Jagannathan, R., Implications of security market data for models of dynamic economies. \textit{J. Polit. Econ.}, 1991, \textbf{99}, 225--262.
\bibitem[Hansen and Richard(1987)]{HansenRichard1987} Hansen, L.P. and Richard, S.F., The role of conditioning information in deducing testable restrictions implied by dynamic asset pricing models. \textit{Econometrica}, 1987, \textbf{55}, 587--613.
\bibitem[Jagannathan and Ma(2003)]{JagannathanMa2003} Jagannathan, R. and Ma, T., Risk reduction in large portfolios: Why imposing the wrong constraints helps. \textit{J. Finance}, 2003, \textbf{58}, 1651--1683.
\bibitem[Kan and Zhou(2007)]{KanZhou2007} Kan, R. and Zhou, G., Optimal portfolio choice with parameter uncertainty. \textit{J. Financ. Quant. Anal.}, 2007, \textbf{42}, 621--656.
\bibitem[Kelly et al.(2019)]{KellyPruittSu2019} Kelly, B., Pruitt, S. and Su, Y., Characteristics are covariances: A unified model of risk and return. \textit{J. Financ. Econ.}, 2019, \textbf{134}, 501--524.
\bibitem[Ledoit and Wolf(2004)]{LedoitWolf2004} Ledoit, O. and Wolf, M., A well-conditioned estimator for large-dimensional covariance matrices. \textit{J. Multivar. Anal.}, 2004, \textbf{88}, 365--411.
\bibitem[Ledoit and Wolf(2022)]{LedoitWolf2022} Ledoit, O. and Wolf, M., Quadratic shrinkage for large covariance matrices. \textit{Bernoulli}, 2022, \textbf{28}, 1519--1543.
\bibitem[Lu and Wang(2025)]{LuWang2025} Lu, W. and Wang, Y., A simple realized factor-based portfolio: improving minimum variance portfolio performance by incorporating low-frequency betas. \textit{Quant. Finance}, 2025, \textbf{25}, 1315--1332.
\bibitem[Markowitz(1952)]{Markowitz1952} Markowitz, H., Portfolio selection. \textit{J. Finance}, 1952, \textbf{7}, 77--91.
\bibitem[Merton(1973)]{Merton1973} Merton, R.C., An intertemporal capital asset pricing model. \textit{Econometrica}, 1973, \textbf{41}, 867--887.
\bibitem[Michaud(1989)]{Michaud1989} Michaud, R.O., The Markowitz optimization enigma: Is `optimized' optimal? \textit{Financ. Anal. J.}, 1989, \textbf{45}, 31--42.
\bibitem[Pearl(1988)]{Pearl1988} Pearl, J., \textit{Probabilistic Reasoning in Intelligent Systems: Networks of Plausible Inference}, 1988 (Morgan Kaufmann: San Mateo, CA).
\bibitem[Pearl(2009)]{Pearl2009} Pearl, J., \textit{Causality: Models, Reasoning, and Inference}, 2nd ed., 2009 (Cambridge University Press: Cambridge).
\bibitem[Pe{\~n}aranda and Sentana(2016)]{PenarandaSentana2016} Pe{\~n}aranda, F. and Sentana, E., Duality in mean-variance frontiers with conditioning information. \textit{J. Empir. Finance}, 2016, \textbf{38}, 762--785.
\bibitem[Pereira(2021)]{Pereira2021} Pereira, C.V., Portfolio efficiency with high-dimensional data as conditioning information. \textit{Int. Rev. Financ. Anal.}, 2021, \textbf{77}, 101811.
\bibitem[Peters et al.(2016)]{PetersBuhlmannMeinshausen2016} Peters, J., B{\"u}hlmann, P. and Meinshausen, N., Causal inference by using invariant prediction: identification and confidence intervals. \textit{J. R. Stat. Soc. B}, 2016, \textbf{78}, 947--1012.
\bibitem[Peters et al.(2017)]{PetersJanzingScholkopf2017} Peters, J., Janzing, D. and Sch{\"o}lkopf, B., \textit{Elements of Causal Inference: Foundations and Learning Algorithms}, 2017 (MIT Press: Cambridge, MA).
\bibitem[Reichenbach(1956)]{Reichenbach1956} Reichenbach, H., \textit{The Direction of Time}, edited by M. Reichenbach, 1956 (University of California Press: Berkeley and Los Angeles).
\bibitem[Rodr{\'i}guez Dom{\'i}nguez(2023)]{RD2023MLWA} Rodr{\'i}guez Dom{\'i}nguez, A., Portfolio optimization based on neural networks sensitivities from assets dynamics respect common drivers. \textit{Mach. Learn. Appl.}, 2023, \textbf{11}, 100447.
\bibitem[Rodr{\'i}guez Dom{\'i}nguez(2025)]{RD2025Causal} Rodr{\'i}guez Dom{\'i}nguez, A., Causal portfolio optimization: principles and sensitivity-based solutions. arXiv:2504.05743, 2025.
\bibitem[Shah and Peters(2020)]{ShahPeters2020} Shah, R.D. and Peters, J., The hardness of conditional independence testing and the generalised covariance measure. \textit{Ann. Stat.}, 2020, \textbf{48}, 1514--1538.
\bibitem[Tiingo(2026)]{Tiingo2026}
Tiingo, Tiingo Daily End-of-Day Prices. Tiingo, 2026.
Available at: \url{https://www.tiingo.com/}.
\bibitem[Wang et al.(2026)]{WangPapailiasVentre2026} Wang, C., Papailias, F. and Ventre, C., Dynamic estimation of sample covariance matrices via hierarchical clustering. \textit{Quant. Finance}, 2026, \textbf{26}, 41--62.\end{thebibliography}
\end{document}